\documentclass{article}
    \usepackage{graphicx}
    \usepackage{amsmath,amssymb,amsthm,mathtools}
    \usepackage{hyperref}
    \usepackage{xspace}
    \usepackage{authblk}
    
    \title{Sedna: Sharding transactions in multiple concurrent proposer blockchains}
    \author[1]{Alejandro Ranchal-Pedrosa}
    \author[1,2]{Benjamin Marsh\thanks{ben@seinetwork.io}}
    \author[3]{Lefteris Kokoris-Kogias}
    \author[3,4]{Alberto Sonnino}
    \affil[1]{Sei Labs}
    \affil[2]{University of Portsmouth}
    \affil[3]{Mysten Labs}
    \affil[4]{University College London}
    \date{December 2025}
    
    \newcommand{\proposal}{Sedna\xspace}
    \usepackage[T1]{fontenc}
\usepackage[dvipsnames]{xcolor} 
\usepackage{amsmath} 
\usepackage{todonotes}
\usepackage[most]{tcolorbox}
\usepackage{graphicx}
\usepackage{placeins}

\newtcolorbox{arnewremarkbox}{
    colback=violet!5!white,    
    colframe=violet!75!black,  
    fonttitle=\bfseries,
    title=A,
    fontupper=\small\itshape,
    before upper={$\blacktriangleright$\space},
    after upper={\space$\triangleleft$},
}

\begin{document}
    \newtheorem{definition}{Definition}
    \newtheorem{assumption}{Assumption}
    \newtheorem{proposition}{Proposition}
    \newtheorem{lemma}{Lemma}
    \newtheorem{corollary}{Corollary}
    \newtheorem{theorem}{Theorem}
    \newtheorem{conjecture}{Conjecture}
    \theoremstyle{remark}
    \newtheorem{remark}{Remark}
    
    \newcommand{\V}{\mathcal{V}}
    \newcommand{\Enc}{\mathsf{Enc}}
    \newcommand{\Dec}{\mathsf{Dec}}
    \newcommand{\AONT}{\mathsf{AONT}}
    \newcommand{\Com}{\mathsf{Com}}
    \newcommand{\E}{\mathsf{E}}
    \newcommand{\Hpub}{H_{\mathrm{pub}}}
    \newcommand{\ID}{\mathsf{txID}}
    \newcommand{\htincl}{\mathrm{ht}_{\mathrm{incl}}}
    \newcommand{\htf}{\mathrm{ht}}
    \newcommand{\shares}{\mathsf{shares}}
    \newcommand{\Share}{\mathsf{Share}}
    \newcommand{\Payload}{\mathsf{payload}}
    
    \newcommand{\R}{\mathsf{R}}
    \newcommand{\Ver}{\mathsf{Ver}}
    \newcommand{\M}{\mathsf{M}}
    \newcommand{\symb}{\ell_{\mathrm{sym}}}
    \newcommand{\epsc}{\varepsilon}
    \newcommand{\delc}{\delta_{\mathrm{code}}}
    \newcommand{\Kneed}{K}
    \newcommand{\advmass}{\alpha}
    \newcommand{\Hpre}{H_{\mathrm{pre}}}

    \newcommand{\PiStr}{\Pi_{\mathrm{str}}}
    \newcommand{\PiMDS}{\Pi_{\mathrm{mds}}}
    \newcommand{\PiRTL}{\Pi_{\mathrm{rtl}}}

    \maketitle
    \begin{abstract}
Modern blockchains increasingly adopt multi-proposer (MCP) consensus to remove single-leader bottlenecks and improve censorship resistance. However, MCP alone does not resolve how users should disseminate transactions to proposers. Today, users either naively replicate full transactions to many proposers, sacrificing goodput and exposing payloads to MEV, or target few proposers and accept weak censorship and latency guarantees. This yields a practical trilemma among censorship resistance, low latency, and reasonable cost (in fees or system goodput).

We present \proposal, a user-facing protocol that replaces naive transaction replication with verifiable, rateless coding. Users privately deliver addressed symbol bundles to subsets of proposers; execution follows a deterministic order once enough symbols are finalized to decode. We prove \proposal guarantees liveness and \emph{until-decode privacy}, significantly reducing MEV exposure. Analytically, the protocol approaches the information-theoretic lower bound for bandwidth overhead, yielding a $2$--$3\times$ efficiency improvement over naive replication. \proposal requires no consensus modifications, enabling incremental deployment.
\end{abstract}
    \section{Introduction}
    
    
    
    Multiple concurrent proposer (MCP) consensus mitigates single leader bottlenecks by allowing many validators to propose per slot. This improves bandwidth utilization and removes the acute impact of a slow or offline leader. Yet MCP alone does not settle how a user should disseminate their transaction to proposers, and \emph{dissemination} plays a key role in the user’s experience of latency, price, and censorship exposure.
    
    An important argument for blockchain systems to resort to MCP instead of single proposer consensus protocols is the native tolerance to censorship by a subset of the proposers. In MCP, it is straight-forward for users to send the same transaction to a number of the proposers and obtain censorship resistance through naive replication, whereas the closest approximation of this approach in single proposer systems involves the user sending the same transaction iteratively to the next proposer until it is included. It is also well-known that this replication comes at a significant impact on goodput: if all users want to deterministically tolerate censorship from $c$ proposers without it incurring a cost on latency, they each need to send their respective transactions to at least $c+1$ validators, incurring a worst-case goodput decrease by the replication factor of $c+1$ (e.g., a decrease of $80\%$ in goodput to tolerate just $5$ censoring proposers).
    
    Some systems enforce deduplication by tuning a global parameter that navigates the trade-off between censorship resistance and goodput for all transactions. Beyond system-wide deduplication efforts, a final trade-off remains: how much of the replication bandwidth cost the system absorbs (i.e., how vulnerable it becomes to DoS attacks on goodput) versus how much users pay themselves (i.e., how costly censorship resistance becomes for them). In this sense, goodput impacts can be equated to pricing considerations for users. It is for this reason that we speak of a \emph{trilemma of three user-facing properties}: censorship resistance (in that the transaction eventually gets included), low latency (in that the transaction gets included as soon as there is capacity for it), and reasonable pricing (either in goodput or economic terms, and whether incurred to the system or user).
    
    \paragraph{MEV and dissemination latency.}
    Some users or systems may tolerate losing one of the three properties for specific transactions, but many applications (e.g., trading) require all three, because latency effectively counts as censorship: a brief delay can wipe out an opportunity even if the transaction eventually lands on-chain.
    Even short delays before inclusion are economically exploitable as MEV. Under MCP, the lever extends from in-block ordering to \emph{who} first learns a transaction and \emph{when}. Replicating broadly increases leakage surface; targeting a single proposer is cheaper but trivially censorable. Our goal is to minimize pre-inclusion leakage while keeping byte cost low.
    

    \paragraph{Our approach.}
We present \proposal, a dissemination protocol that replaces whole-transaction replication with \emph{verifiable, ratelessly coded symbols} addressed to a subset of proposers (``lanes''). The sender commits to the payload, derives a transaction identifier~$\ID$, and generates an unbounded stream of small coded symbols, each tied to the transaction via signatures. 
These symbols are packaged into \emph{addressed bundles} and privately delivered to a sampled set of lanes. Inclusion occurs as soon as enough \emph{distinct verified symbols} for~$\ID$ appear in finalized history to reconstruct the payload; at that point the payload is decoded and the transaction is executed. Order is \emph{deterministic}: transactions are sorted by~$(\htincl(\ID),\ID)$, where $\htincl(\ID)$ is the first height at which decoding succeeds and the commitment opening verifies. Crucially, \proposal is entirely user-facing: each sender independently selects encoding parameters to navigate the trilemma according to their transaction's specific requirements, without any protocol-level changes to the underlying MCP consensus.


    \paragraph{Benefits.}
For large payloads, the byte overhead of \proposal approaches
$\frac{1+\varepsilon}{1-c_e/n},$
where $\varepsilon>0$ is the small coding overhead of the rateless scheme (e.g., $5\%$), $n$ is the number of proposers, and $c_e$ is the \emph{effective} number of censoring proposers the user wishes to tolerate. This matches the information-theoretic lower bound for any deterministic coding scheme up to the $(1+\varepsilon)$ factor, and is far below the $m$-fold cost of naive replication. Our evaluation demonstrates that for medium-to-large payloads, \proposal reduces bandwidth overhead by a factor of $2$--$3\times$ compared to naive replication under typical censorship assumptions, with the gap widening as payload size increases. A user selects how many lanes to contact and how many symbols to place in each bundle so as to meet a target per-slot failure probability $\delta$. Because symbols are verifiable and payload-hiding until decode, \proposal provides \emph{until-decode privacy} such the adversary's early decode probability remains near zero for properly chosen parameters, substantially reducing the pre-inclusion MEV surface compared to transparent mempool dissemination. For completeness, we also analyze two simpler submission variants, naive replication and fixed-rate MDS coding, and identify the parameter regimes where each approach is bandwidth-optimal.

    
    \paragraph{Contributions.}
    \begin{itemize}
      \item \textbf{Protocol.} We introduce \proposal, a user-facing submission layer for MCP that replaces whole-transaction replication with commitment-bound, addressed bundles of verifiable coded symbols; it realizes decode-to-include semantics and a deterministic execution order in a lazy-execution setting, while retaining pay-for-bytes admission.
\item \textbf{Correctness.} We prove header/commitment non-malleability and deterministic resolution of symbol indices, and show that monotone decoding yields a unique payload from any \(K\) verified symbols for honestly encoded transactions (where \(K\) is the decode threshold) except with negligible decoding failure probability \(\delc\).
      \item \textbf{Liveness and censorship resistance.} We derive per-slot inclusion bounds and give closed-form, conservative prescriptions for how many lanes to contact (and how many symbols to place per bundle) to achieve a target failure probability \(\delta\) against an assumed effective censor mass \(c_e\).
      \item \textbf{Privacy (and MEV).} We formalize “until‑decode’’ confidentiality, proving that pre‑decode leakage is bounded to the bits contained in revealed symbols and the public header, and we quantify an adversary’s early‑decode probability (including passive eavesdropping). We discuss how these bounds reduce the pre‑inclusion MEV surface and interact with deterministic ordering by \((\htincl(\ID),\ID)\).
      \item \textbf{Cost and comparisons.} We compute exact byte costs including header and per-symbol metadata, derive asymptotic overhead floors approaching \(\tfrac{1+\varepsilon}{1-c_e/n}\), and compare \proposal\ to naive replication and fixed-rate MDS coding, identifying parameter regimes where each approach is bandwidth-optimal, both theoretically and empirically.
    \end{itemize}

    \paragraph{Roadmap.}
    Section~\ref{sec:background} explores the background and related work. Section~\ref{sec:model} introduces the system model.
    Section~\ref{sec:protocol} specifies \proposal.
    Section~\ref{sec:guarantees} proves safety, liveness, and privacy.
    Section~\ref{sec:performance} analyzes byte overhead and parameter selection,
    and Section~\ref{sec:eva} evaluates representative design points. We finally conclude in Section~\ref{sec:conclusion}.
    
    \section{Background and Related Work}
    \label{sec:background}

Recent work establishes that the latency cost of censorship resistance is fundamental: any protocol achieving this guarantee requires $4$ rounds in synchrony and $5$ rounds in partial synchrony, exactly two rounds more than classic leader-based Byzantine Broadcast~\cite{abraham2025latency}. The MCP framework~\cite{garimidi2025mcp} achieves these bounds, making it the canonical architecture for blockchains that take censorship resistance seriously. Our work complements this foundation by addressing the \emph{bandwidth} dimension of the problem: given that MCP-style dissemination is required for strong censorship guarantees, how can users minimize the byte overhead of ensuring their transactions reach sufficiently many proposers?

    \paragraph{MCP approaches.}
    
    A large class of protocols finalize \emph{sets} or vectors of proposals per round or slot and are therefore natural instances of MCP \cite{giridharan2024autobahn, miller2016honeybadger,keidar2021dagrider,danezis2022narwhal,kokoriskogias2022bullshark,gagol2019aleph,baird2016hashgraph}. 
     MCP reduces the \emph{single-leader veto} because multiple proposers can carry a user’s data in a slot. Besides a mention to a direct trade-off between deduplication and censorship resistance~\cite{danezis2022narwhal, stathakopoulou2019mir} through naive replication, these MCP protocols mainly focus on consensus properties assuming transactions are already in the proposers’ mempools; they do not natively optimize the user‑side byte cost of reaching many lanes, nor do they provide quantitative knobs for per‑slot inclusion probability under an assumed censor mass \cite{miller2016honeybadger,danezis2022narwhal,kokoriskogias2022bullshark}. 
    For instance, transactions in the Sui blockchain~\cite{sui2022whitepaper} are naturally duplicated up to 5 times~\cite{sui-code}, which reduces the system’s goodput by a significant factor. Moreover, the system offers no principled mechanism that prevents a user from submitting the same transaction to all validators, and such a submission can cause the transaction to be processed many times.
    
     \paragraph{MCP in production.} Nevertheless, the trilemma is unavoidable, with MCP systems in production taking different approaches to tackle it. The most prevalent approach implies relying on trusted RPC servers relaying transactions in a way that it enforces the deduplication target of the system. This approach, besides being relying on trusted, centralized entities, does not entirely solve the problem of resolving who is to absorb the cost of duplication, users or the system. Sui, Aptos, Red-Belly decide to absorb potential duplications as cost to the system. Hedera instead charges this cost to users, that need to pay for the total cost of any duplication of their transaction, which places the trilemma between latency, censorship resistance and cost in user terms and not system's goodput terms. It is also worth noting that Hedera cannot guarantee full protection from incurred costs on the systems due to the asynchronous execution model, in that a user may incur a cost greater than the funds the system can take from its accounts. Filecoin, also technically an MCP system, does not rely on RPC servers and remains trustless, with a replication factor average ~3-4, and some transactions being duplicated as much as 9 times\footnote{\url{Source: https://filfox.info/en/message/bafy2bzacebyo7g3tugujmshqnxwlqvp3qxl2omxk4bol6arp4ah76lme6nyzg?t=1}}, whose cost is entirely absorbed by the system.
    
    One could argue that some MCP designs pursue alternative strategies to sidestep the trilemma while preserving privacy and low latency. However, in nearly all deployed MCP variants, transparent mempools expose transactions before proposal, which enables pre-inclusion leakage and facilitates classic MEV behaviors such as front-running, sandwich attacks, and strategic reordering \cite{daian2020flashboys}. In effect, MCP enhances the \emph{opportunity} for inclusion, but current implementations still lack a user-facing, bandwidth-aware submission layer that enforces privacy until decode and offers explicitly tunable liveness guarantees.
    
    
    
    \paragraph{Single‑leader systems.}
    
   Single-leader, per-slot architectures concentrate inclusion authority in a single leader for each slot \cite{buterin2020eth2,gabizon2018hotstuff}. To mitigate MEV, Ethereum has deployed a \emph{Proposer-Builder Separation} through MEV-Boost and private relay markets, and is actively investigating in-protocol PBS mechanisms \cite{flashbotsMevBoost,ethereumPBS}. Complementary private order-flow channels, such as Flashbots Protect, bloXroute’s private transactions, and MEV-Blocker, postpone transaction disclosure until the builder phase \cite{flashbotsProtect,bloxroutePrivateTx,mevBlocker}. In parallel, threshold and timelock encryption schemes seek to conceal transaction contents until after ordering completes \cite{shutter2021whitepaper}. 
    
    These mitigations rely on off‑protocol markets and trusted relay behaviors, complicate decentralization (builder concentration), and do not resolve the fundamental \emph{single inclusion bottleneck}: a leader can still censor or strategically delay \cite{flashbotsMevBoost,daian2020flashboys,ethereumPBS}. Moreover, private relays reduce but do not eliminate leakage (e.g., side‑channels, orderflow auctions) and provide no user‑level control over bandwidth amortization across multiple inclusion opportunities.
    
    \paragraph{Order-fairness and encryption.}
    
    Two main lines appear beyond PBS/relays. First, \emph{order‑fairness} protocols (Aequitas~\cite{kelkar2020aequitas}, Themis~\cite{kelkar2023themis}, SpeedyFair~\cite{mu2024speedyfair}) attempt to align ordering with network timing to reduce extractable value from reordering. Second, \emph{encryption‑based} schemes~\cite{rivaseahorse,bormet2025beat,choudhuri2024mempool} hide contents (threshold/timelock encryption, TEEs), often trading off latency/complexity and introducing trust or liveness assumptions \cite{shutter2021whitepaper}. Separately, batch‑auction designs (e.g., frequent batch auctions and on‑chain batch DEXs) reduce the value of marginal reordering \cite{budish2015fba,cow2021}. None of these lines simultaneously addresses (i) the user’s byte‑budget for multi‑lane dissemination, (ii) quantitative, per‑slot liveness against a target censor mass, and (iii) deterministic post‑decode execution semantics.
    
    \paragraph{Verifiable information dispersal.}

Asynchronous Verifiable Information Dispersal (AVID)~\cite{cachin2005avid}
allows a dealer to disperse data to $n$ parties such that any $t+1$ honest
parties can reconstruct, while guaranteeing the dealer cannot equivocate. In
the VSS/VID literature this is often phrased as a \emph{binding} property: by
the time the first honest party completes the dissemination protocol, the
dealer is \emph{bound} to a single value, and cannot later make reconstruction
yield a different (even invalid) value depending on additional information such
as randomness or subsequent blocks.

While AVID provides binding and consistency, meaning every reconstruction
yields the same output, it does not ensure that the dispersed data carries
semantic meaning. Furthermore, classic AVID schemes are overkill for our
setting as they require the dealer to commit to the entire data bundle
beforehand. \proposal\ relaxes this constraint while maintaining a VID-style
binding guarantee through signature-based bundle verification, per-index
deduplication, and a deterministic decoding rule defined over the on-chain
total order of symbols.

Crucially, we use the consensus output to determine which subset of data is
used for decoding, allowing the sender to continue encoding and committing to
new symbols on the fly. This is particularly useful for the rateless encoding
version of \proposal, where the sender can dynamically adapt the number of lanes
used based on how quickly initial symbols are included on-chain. This
flexibility allows \proposal to dynamically balance the required latency, fees,
and censorship resistance. We intentionally avoid proving that dispersed data
is well formed; instead, we replace AVID-style correctness proofs with economic
deterrence, as malicious senders incur bandwidth fees for every published
bundle.

    \paragraph{Replication approaches.}
    At the network layer, Reed–Solomon/MDS codes split a payload of size \(S\) into \(m\) shares such that any \(k\) suffice to reconstruct \cite{reed1960polynomial,wicker1995rs}. In practice, high-throughput deployments and recent prototypes apply erasure-coded, tree-based or multipath fanout at the consensus dissemination layer to accelerate block propagation and improve robustness under partial synchrony \cite{solanaTurbine,chan2023simplex, shoup2023sing, kniep2025solana, danezis2025walrus}. In contrast, LT/Raptor/RaptorQ codes generate an unbounded stream of small symbols with near-linear decoding and tiny overhead \((1+\varepsilon)\) \cite{luby2002lt,shokrollahi2006raptor,rfc6330,rfc5053}. 
    
    In data availability systems, erasure coding and sampling are typically used for \emph{post-proposal} data availability rather than user-side dissemination to multiple proposers \cite{albassam2021lazyledger,eip4844, danezis2025walrus, goren2025shelby}. 
    
    Overall, these mechanisms are promising but operate at the system's consensus dissemination or data availability layers; and they are thus orthogonal to the user-facing submission protocol we propose with \proposal.

    \section{System Model and Preliminaries}
    \label{sec:model}
    
    \paragraph{Consensus.} We fix an integer $n\ge 3$ and a validator set $\V=\{1,\dots,n\}$. Time proceeds in slots $t\in \mathbb{N}$. In each slot, every validator $i\in\V$ proposes a block $B_i^t$; a vector-commit consensus abstraction finalizes, at slot height $t$, an ordered $n$-tuple $(B_1^t,\ldots,B_n^t)$. We refer to the conceptual stream of proposals from validator $i$ as a lane. We assume execution is \emph{lazy} that is to say that the system commits to data and order first, and only executes transactions after finality has been reached.
    
    \paragraph{Network.} We assume a partially synchronous network model, with the adversary initially being able to delay messages arbitrarily. There is an unknown Global Stabilization Time (GST) after which all messages sent by honest parties are guaranteed to be delivered by honest parties within a fixed bound $\Delta$. The liveness of the underlying vector-consensus abstraction (i.e., its ability to finalize new blocks) depends on this assumption. Safety properties hold at all times, while liveness properties are guaranteed to hold after GST. Our analysis assumes that bundles have already been delivered to target validators' mempools; the mechanics of user-to-validator dissemination are orthogonal to the protocol's core guarantees.
    
    \paragraph{Good case.}
    We additionally reason about particular properties of inclusion of a transaction submitted by an honest user post-GST. In this ``good case'', to isolate consensus-layer liveness, we assume that the user's bundles (sent at or before $t-\Delta$) have been successfully delivered to the target validators' mempools, and that validators have sufficient blockspace for all delivered, valid bundles. This removes fee-market and block-packing effects from the analysis, allowing us to evaluate the protocol's intrinsic properties of censorship resistance, low latency, and reduced duplication factor. These assumptions are not required for correctness.

    \paragraph{Transaction format.}
A transaction is a pair $tx=(\Hpub,\Payload)$, where $\Payload\in\{0,1\}^{S_{\mathrm{pl}}}$ is the payload and $\Hpub$ is a public header carrying the fee/accounting metadata together with a commitment to the payload. The sender samples commitment randomness~$\sigma$ and computes a commitment
\[
  C = \Com(\sigma,\Payload).
\]
Let $M := (\sigma \parallel \Payload)$ be the message dispersed by the code, of total length
\[
  S := |\sigma| + S_{\mathrm{pl}}.
\]
The value $\sigma$ is included as a prefix of the data to be encoded, i.e., the encoder operates on $(\sigma \parallel \Payload)$. We fold $|\sigma|$ into $S$ for notational simplicity. This ensures $\sigma$ becomes available upon successful decoding, enabling validators to verify the commitment opening.
For the remainder of the paper we abuse notation slightly and refer to $S$ as the
“payload size”, since $|\sigma|$ is fixed and negligible compared to $S_{\mathrm{pl}}$
in our parameter regimes. We write
\[
  \Hpub = (\Hpre,C,\Sigma),\qquad
  \ID = H(\Hpre \parallel C),\qquad
  \Sigma = \mathrm{Sig}_{\mathsf{sk}}(\ID).
\]
We call $\Hpre$ the \emph{preimage header}: a deterministically serialized tuple of
fee/accounting fields required for mempool admission and bandwidth pricing.
    \paragraph{MDS coding.} Our \proposal protocol comes in three different variants depending on the technique used for replication: naive, MDS and rateless. As a baseline fixed-rate scheme, we consider $(m, k)$ Maximum Distance Separable (MDS) codes. The sender encodes the payload $S$ into $m$ shares, each approximately $S/k$ in size. Successful reconstruction of the original payload requires the recipient to gather any $k$ distinct shares.
    \paragraph{Rateless, verifiable coding.}
The main variant of our protocol employs a \emph{rateless, verifiable encoder}~$\R$.
On input the message $M := (\sigma \parallel \Payload)$, $\R$ produces an unbounded
sequence of coded symbols
\[
  (y_1, y_2, \ldots), \qquad y_j \in \{0,1\}^{\symb}.
\]
For clarity, we write the $j$-th symbol generator induced by~$\R$ as~$R_j$.
Let $M \in \{0,1\}^S$ denote the (fixed) input message. Each $R_j$ is a deterministic function
\[
  R_j : \{0,1\}^* \to \{0,1\}^{\ell_{\mathrm{sym}}},
\]
such that $y_j = R_j(M)$,
i.e., $y_j$ is the $j$-th coded symbol produced by~$\R$ (with any required
per-symbol coefficient information derived deterministically from $j$ and public parameters).

For a fixed transaction identifier $\ID$, any verified pair $(j,y_j)$ that appears in a bundle with $\Ver(\ID,i,J_i,\{y_j\}_{j\in J_i},\Sigma_i)=1$ is cryptographically bound to that transaction via the sender's signature.


We additionally assume that for any fixed $\Payload$, the collection of symbols $\{y_j\}$ behaves as (pseudo)random linear combinations of the payload blocks, so that any set of $r$ symbols reveals at most $r\,\ell_{\mathrm{sym}}$ bits of information about~$\Payload$. This property underpins our privacy analysis in Section~\ref{sec:guarantees}.

    From any $\Kneed$ verified distinct symbols (for indices $j$ of the verifier's choice),
$\Dec$ reconstructs $M = (\sigma, \Payload)$ with probability at least $1-\delc$. We parameterize
    \[
      \Kneed \ \stackrel{\mathrm{def}}{=}\ \left\lceil (1+\epsc)\,\frac{S}{\symb}\right\rceil ,
    \]
    where $\epsc>0$ is a small overhead (e.g., $5\%$) and $\delc$ the residual decoding failure probability under the chosen rateless code family (e.g., LT/Raptor/RaptorQ).

    \paragraph{Sender model.} Senders can customize their desired trade-off between latency, cost, and censorship-resistance for each transaction. Sender $s \in \mathcal{S}$ has a specific censorship-resistance requirement for each transaction $tx_s$, expressed in the number of censoring validators $c(tx_s)$ the inclusion of transaction $tx_s$ must tolerate. In particular, it is clear that at least $c(tx_s)+1$ validators must be reached, $0\leq c(tx_s)\leq n-1$, to deterministically tolerate censoring from up to $c(tx_s)$ validators. We abuse notation by referring to $c$ in the remainder of this document. Intuitively, for the naive approach of re-submitting the full transaction payload and metadata to $c$ different validators, $c+1$ becomes the replication factor of the transaction. Thus, while $f$ is a system parameter of Byzantine fault tolerance, $c$ is a user parameter of censorship resistance. 
    
    \paragraph{Effective censorship.} Given $f, c$ we define $c_{e}(f,c)$ as the effective number of censored lanes. We abuse notation by referring to $c_{e}$. Typically, for $c\leq f$ then the adversary is not strong enough to selectively censor particular lanes, meaning that it can only censor the transaction in the lanes that it controls directly, hence $c_e=c$. For $c>f$, the adversary can additionally prevent progress on correct lanes, as finalization requires $\lfloor(n+f)/2\rfloor+1=2f+1$ votes (in the $n=3f+1$ model), but since it is however possible for all protocols to enforce at least $n-f$ lanes must be decided in any iteration of MCP, then $c_e=n-(2f+1)+c$. Note that in this case $c_e=n$ for $c>2f$ and thus the strongest adversary tolerable is $c=2f$. Nonetheless, as the concrete power of the adversary to censor lanes outside of its direct control is dependent on the particular instance of the consensus implementation, we simply refer to $c_e$ in the remainder of this document.
    
    \section{\proposal Protocol}
    \label{sec:protocol}
    We now present the \proposal protocol. At a high level, the sender (1) commits to the transaction payload and derives a unique identifier $\ID$, (2) generates an unbounded stream of small, verifiable coded symbols, (3) packages disjoint subsets of symbols into bundles addressed to a sampled set of validator lanes, and (4) privately delivers these bundles. Validators verify incoming bundles by checking hash consistency, the header signature, the local accounting predicate, and the bundle signature. They do \emph{not} verify that symbols are correctly encoded with respect to the commitment; malformed or inconsistent symbols are handled economically (fees are charged) and filtered by post-decode transaction validity checks.

    \paragraph{Commitment and addressed shares.}
Using the commitment randomness $\sigma$ and commitment $C=\Com(\sigma,\Payload)$ from Section~\ref{sec:model}, the sender obtains the identifier $\ID = H(\Hpre \parallel C)$ and header $\Hpub=(\Hpre,C,\Sigma)$. For a chosen set of symbol indices $j$, the sender forms addressed shares
\[
  \Share_{i,j}=(\ID,i,j,y_j,\Hpub).
\]
    Validators maintain coded mempools of addressed bundles that passed local checks and bytes-fee payment. A block $B_i^t$ is a sequence of such addressed bundles assembled by validator $i$.

    \proposal may prescribe more than one share to be sent to the same validator. As a result, we aggregate a set of shares with the same validator $v$ as recipient $\{\Share_{i,j}\}_{i=v}$ into a bundle:
    \[
      Bundle_i=\bigl(\ID,\,i,\,J_i,\,\{y_j\}_{j\in J_i},\,\Sigma_i,\,\Hpub\bigr),
    \]
    where $J_i$ is the set of symbol indices carried in the bundle and
\[
  \Sigma_i = \mathrm{Sig}_{\mathsf{sk}}\!\bigl(\ID \parallel i \parallel \langle j, y_j\rangle_{j \in J_i}\bigr)
\]
is the sender's signature binding all symbols $y_j$ to their indices $j$ and lane $i$ for transaction $\ID$, where $\langle j, y_j \rangle_{j \in J_i}$ denotes the list of index–symbol pairs sorted by increasing $j$.

For a bundle addressed to lane $i$ with index set $J_i$ and symbols $\{y_j\}_{j\in J_i}$, we write
\[
  \Ver\!\bigl(\ID,i,J_i,\{y_j\}_{j\in J_i},\Sigma_i\bigr)=1
\]
to mean that the bundle signature verifies, i.e.,
\[
  \mathrm{Verify}_{\mathsf{pk}}\!\bigl(\ID \parallel i \parallel \langle j, y_j\rangle_{j \in J_i}, \Sigma_i\bigr) = 1.
\]
For an honest sender and any external adversary that does not know $\mathsf{sk}$, EUF-CMA security implies that any bundle with $\Ver=1$ must have been explicitly produced by that sender; an external adversary cannot alter indices or symbol values without invalidating the signature.


    
\paragraph{Verifying shares/bundles.}
Upon receiving a bundle addressed to lane $i$ with index set $J_i$ and symbols $\{y_j\}_{j\in J_i}$, a validator performs the following checks:
\begin{enumerate}
  \item \emph{Hash consistency:} recompute $\ID' := H(\Hpre\parallel C)$ from the header and check that $\ID'=\ID$.
  \item \emph{Header signature:} $\mathrm{Verify}_{\mathsf{pk}}(\ID,\Sigma)=1$.
  \item \emph{Accounting/fees:} the admission and bandwidth-pricing predicate on $\Hpre$ is satisfied.
  \item \emph{Bundle signature:} $\mathrm{Verify}_{\mathsf{pk}}(\ID \parallel i \parallel \langle j, y_j\rangle_{j \in J_i}, \Sigma_i)=1$, i.e., $\Ver(\ID,i,J_i,\{y_j\}_{j\in J_i},\Sigma_i)=1$ as in Definition~\ref{def:ver-relation}.
\end{enumerate}
A bundle is admitted to the mempool only if all four checks pass.


    \paragraph{Cross-bundle equivocation.}
Bundle-level signatures bind symbols within a single bundle but do not cryptographically prevent a malicious sender from signing different values for the same index $j$ in bundles addressed to different lanes. However, this equivocation is harmless: per-index deduplication (Section~\ref{sec:guarantees}) ensures that only the first verified symbol for each $(\ID, j)$ pair in finalized history is considered, so all honest validators agree on a unique value for each index. Moreover, if a sender equivocates, the resulting symbol set is very unlikely to correspond to a meaningful payload: decoding will typically either fail or yield an invalid transaction, which is discarded while bandwidth fees remain charged. In either case, equivocation is economically irrational.

\paragraph{Deduplication order.}
The finalized ledger at height $h$ consists of an ordered sequence of block vectors $(B_1^t, \ldots, B_n^t)$ for $t \le h$. We scan this sequence in increasing height $t$, then by lane index $i \in \{1, \ldots, n\}$, then by intra-block position. The first occurrence of any $(\ID, j)$ in this order determines the unique symbol value $y_j$ used for decoding; subsequent occurrences are ignored.

    
    
    
    
    
    \paragraph{Lane sampling.}
    For each transaction, the sender samples a subset $U\subseteq \V$ of size $0<m\leq n$ and privately delivers one addressed bundle to each lane in $U$. For an honest sender, the index sets $J_i$ across lanes are chosen to be disjoint, so that each published bundle
    contributes $s$ previously unseen indices.

  \paragraph{Inclusion, decoding, and execution order.}
For a fixed identifier $\ID$, let $X_{\ID}(h)$ denote the set of deduplicated verified index–symbol pairs $(j,y_j)$ for $\ID$ that appear in finalized history up to height $h$. Decoding is attempted whenever $|X_{\ID}(h)|\ge \Kneed$. Upon successful decoding, validators recover $(\sigma, \Payload)$ and verify the commitment opening:
\[
  C \stackrel{?}{=} \Com(\sigma, \Payload).
\]
If verification succeeds, the transaction is included. We define the \emph{inclusion height}, denoted $\htincl(\ID)$, as the minimal block height $h$ at which both decoding succeeds and the commitment opening verifies. The final, deterministic execution order is then established by sorting all included transactions lexicographically based on the tuple $(\htincl(\ID), \ID)$.

If a decoding attempt fails at some $h$ (an event of probability at most $\delc$ for an honestly encoded transaction once $|X_{\ID}(h)|\ge \Kneed$), the protocol waits for additional verified symbols and retries at the next height where new symbols are available. A transaction that fails post-decode verification (e.g., commitment opening fails, or application-level validity checks fail) is discarded; bandwidth fees remain charged.

    
    \section{\proposal Guarantees}
    \label{sec:guarantees}
    
    We now establish the main correctness, liveness, and privacy properties of \proposal. Unless stated otherwise, we work post-GST under the network and consensus assumptions of Section~\ref{sec:model}, and we consider the rateless
    variant defined in Section~\ref{sec:protocol}. 
    Throughout this section we assume that the commitment scheme $\Com$ is 
computationally binding and hiding, that each symbol generator $R_j$ is 
deterministic, that the hash function $H$ is collision resistant, and that the signature scheme is EUF-CMA secure.
    
    
    We first formalize the verification relation for bundles and the decoder requirements.
    \begin{definition}[Verification relation]
\label{def:ver-relation}
For a bundle addressed to lane $i$ with index set $J_i$ and symbols $\{y_j\}_{j\in J_i}$, we write $\Ver(\ID,i,J_i,\{y_j\}_{j\in J_i},\Sigma_i)=1$ iff the bundle signature verifies:
\[
  \mathrm{Verify}_{\mathsf{pk}}\!\bigl(\ID \parallel i \parallel \langle j, y_j\rangle_{j \in J_i}, \Sigma_i\bigr) = 1,
\]
where $\langle j, y_j \rangle_{j \in J_i}$ denotes index–symbol pairs sorted by increasing $j$. By EUF-CMA security, any bundle with $\Ver=1$ was produced by the holder of $\mathsf{sk}$.
\end{definition}
    \begin{definition}[Inclusion Height]
\label{def:incl-height}
For a fixed transaction identifier $\ID$ and ledger height $h$, let
\[
X_{\ID}(h)
  = \Bigl\{
      (j, y_j) \ \Big|\ 
      \begin{aligned}[t]
        &\exists\,\text{a finalized bundle }B\text{ with }\htf(B)\le h\\
        &\text{whose deduplicated contribution for }(\ID,j)\text{ is }y_j
      \end{aligned}
    \Bigr\}.
\]
Decoding is attempted whenever $|X_{\ID}(h)| \ge K$.  
The inclusion height of $\ID$, denoted $\htincl(\ID)$, is
\[
  \htincl(\ID)
    \;=\;
    \min\Bigl\{\, h \ \Big|\ 
    \begin{aligned}[t]
      &|X_{\ID}(h)| \ge K,\ 
        \Dec(X_{\ID}(h)) = (\sigma, \Payload),\\
      &\text{and } C = \Com(\sigma, \Payload)
    \end{aligned}
    \Bigr\}.
\]
If decoding fails or commitment verification fails at height $h$, the system waits for additional verified symbols and retries at the next height where new symbols are available.
\end{definition}

    \begin{definition}[Monotone decoding for honest encodings]
\label{def:monotone-dec}
The decoder $\Dec$ is \emph{monotone for honestly encoded transactions} if, for any two index sets $J_1,J_2$ with $|J_1|,|J_2|\ge K$, when the corresponding symbols are generated by the same honest sender as $y_j=R_j(M)$, either $\Dec$ outputs the same $M$ (containing payload) on both sets, or $\Dec$ fails on at least one. Any residual disagreement probability is absorbed into~$\delc$.
\end{definition}
    
We work under the validity and deduplication rules of Section~\ref{sec:protocol}. Recall that deduplication affects symbol selection, not fee accounting: duplicate bundles that reach finalized history still incur bandwidth fees.
Note that verification does not enforce correct encoding; a malicious sender may sign arbitrary symbols. Garbage payloads are deterred economically, as fees are charged regardless of whether decoding succeeds.

    
    \subsection{Safety and Correctness}
We establish the core safety properties of \proposal: that for any finalized ledger the set of symbols used for each index is deterministically resolved, that headers and commitments cannot be malleated, and that decoding yields a unique payload. Together these ensure that the execution order $(\htincl(\ID),\ID)$ is a deterministic function of the finalized ledger.

\begin{lemma}[Validity and Deterministic Resolution]
\label{lem:validity}
If a finalized block contains a bundle with $\Ver=0$, the block is invalid and must be rejected. Moreover, while a malicious sender may produce conflicting verified bundles for the same $\ID$, the protocol's deduplication rule ensures that for any fixed finalized ledger, the set of accepted symbols is uniquely determined.
\end{lemma}
\begin{proof}
Consensus validity requires $\Ver=1$ for all included bundles. Conflicting symbols for the same $(\ID, j)$ are resolved by per-index deduplication: only the first occurrence in finalized order is used.
\end{proof}

    
    
    \begin{lemma}[Header and Commitment Non-Malleability]
    \label{lem:header-nm}
    Let $\ID = H(\Hpre\parallel C)$ and $\Sigma=\mathsf{Sig}_{\mathsf{sk}}(\ID)$. Assume $H$ is collision resistant and the signature scheme is EUF-CMA secure. Then, given a valid pair $(\Hpre,C)$ and $(\ID,\Sigma)$, no PPT adversary can produce $(\Hpre',C')\neq(\Hpre,C)$ with the same identifier $\ID$ and a valid signature under the sender's key, except with negligible probability.
    \end{lemma}
    
    \begin{proof}
    If $(\Hpre',C')\neq(\Hpre,C)$ and $H(\Hpre'\parallel C') = H(\Hpre\parallel C)$, then $H$ has a collision. If the
    adversary instead changes $\Hpre'$ but reuses $\Sigma$ on the same $\ID$, it must either break the binding between $(\Hpre,C)$ and $\ID$ (again causing a hash collision) or forge a signature under the sender's key, violating
    EUF‑CMA security.
    \end{proof}

\begin{lemma}[Unique decode for honest senders]
\label{lem:soundness}
Assume the sender is honest and encodes as $y_j = R_j(M)$ for all $j$,
where $M = (\sigma \parallel \Payload)$.
Assume $\Dec$ outputs $M$ from any set of at least $K$ correctly encoded index-symbol pairs with probability at least $1-\delc$.
Then any two sets $X_1,X_2$ of correctly encoded index–symbol pairs with $|X_1|,|X_2|\ge K$ decode to the same payload,
except with probability at most $\delc$.
\end{lemma}
\begin{proof}
Each $y_j = R_j(M)$ is unique by determinism. By Lemma~\ref{lem:validity}, deduplication fixes a unique symbol per index. Thus any two sets of $K$ verified symbols are subsets of the same codeword, decoding to the same payload except with probability $\delc$.
\end{proof}

For a malicious sender who signs inconsistent or garbage symbols, decoding may fail or produce an invalid transaction. In this case, the transaction is discarded but bandwidth fees (prepaid in $\Hpub$) remain charged. This fee mechanism deters garbage submissions without requiring proofs of correct encoding.

    
    For a fixed $\ID$ and height $h$, we write $X_{\ID}(h)$ for the set of distinct verified index–symbol pairs as in Definition~\ref{def:incl-height}. Per-index deduplication ensures that duplicates do not change $X_{\ID}(h)$.
    
    \begin{theorem}[Order Determinism]
    \label{thm:order}
    Let $\htincl(\ID)$ be the inclusion height of~$\ID$ as in Definition~\ref{def:incl-height}. Then the execution order obtained by sorting all included transactions lexicographically by $(\htincl(\ID),\ID)$ is a deterministic function of the finalized ledger.
    \end{theorem}
    \begin{proof}
        Finality fixes the ledger prefix up to height $h$, hence fixes $X_{\ID}(h)$ for every $\ID$ and $h$ (Lemma~\ref{lem:validity}).
Since $\Dec$ is deterministic, its success/failure and output on input $X_{\ID}(h)$ are fixed (Lemma~\ref{lem:soundness}).
Therefore $\htincl(\ID)$ is fixed by the finalized ledger, and so is the lexicographic order on $(\htincl(\ID),\ID)$.

    \end{proof}


    
    \begin{theorem}[Non-Malleability of \proposal]
\label{thm:nm}
Fix a finalized ledger. Under the assumptions above, for each identifier $\ID$ corresponding to an honestly encoded transaction, the decoded payload and its execution position $(\htincl(\ID),\ID)$ are uniquely determined by that ledger, except with probability at most~$\delc$ due to decoder failure.
\end{theorem}
\begin{proof} 
By definition, $X_{\ID}(h)$ depends only on the finalized prefix up to $h$. By Lemma~\ref{lem:validity}, symbols are deterministically resolved; by Lemma~\ref{lem:soundness}, decoding is unique; by Theorem~\ref{thm:order}, execution position is fixed. The only failure mode is decoding failure (probability at most $\delc$).
\end{proof}

    Theorem~\ref{thm:nm} constrains semantics for a fixed finalized history; it does not constrain an adversary’s influence on $\htincl(\ID)$ via censorship or scheduling. We prove censorship resistance and liveness in Section~\ref{sec:liv}.
    
    \subsection{Liveness and Censorship Resistance}
    \label{sec:liv}
    We now quantify the per-slot success probability of inclusion under the good case regime. Let $m$ denote the number of addressed lanes for a transaction, let $s$ denote the number of symbols per bundle, and let $c_e$ be the effective number of censored lanes as defined in Section~\ref{sec:model}. Let $H$ be the random variable counting how many of the $m$ addressed lanes are
    honest. Since we are sampling without replacement from the $n$ lanes, we have
    \[
      H \sim \mathrm{Hypergeom}(n,n-c_e,m).
    \]

    \begin{lemma}[Single-slot Inclusion (good case)]
    \label{lem:liveness-one}
    Assume in a slot all addressed honest lanes publish their bundles. Then inclusion succeeds in that slot with probability at least
    \[
      (1-\delc)\cdot \Pr\!\big[\,H \ge \lceil K/s\rceil\,\big].
    \]
    \end{lemma}
    
    \begin{proof}
    Condition on the event $\{H \ge \lceil \Kneed/s\rceil\}$. Each honest lane contributes $s$ verified symbols, so at least $\Kneed$ distinct verified indices for~$\ID$ appear in finalized history. By the decoder guarantee, decoding succeeds with probability at least $1-\delc$ given any $\Kneed$ verified symbols. Taking expectations over the distribution of $H$ yields the stated lower bound. No independence assumption beyond the hypergeometric sampling is required.
    \end{proof}
    
    \begin{lemma}[Multi-slot inclusion under resampling]
    \label{lem:multi-slot-resampling}
    If the sender resamples $U$ each slot from the same distribution (parameters $m,s,K$) and, post-GST, the per-slot success probability satisfies
    \[
    p \ \ge\ (1-\delta_{\mathrm{code}})\cdot \Pr\!\bigl[H \ge \lceil K/s \rceil\bigr],
    \]
    then the number of slots to inclusion is stochastically dominated by a geometric r.v.\ with parameter $p$.
    \end{lemma}
    
    \begin{proof}
    Each slot is a Bernoulli trial that succeeds whenever enough honest addressed lanes are hit and $\Dec$ succeeds. The same conditional argument as in Lemma~\ref{lem:liveness-one} yields the per-slot success probability lower bound, geometric domination follows.
    \end{proof}
    We show in Section~\ref{sec:opt-sub}, Theorem~\ref{thm:rateless_optimal_prob} the optimal submission strategy for a total per-slot failure probability $\delta$, comparing it with variants of naive replication (Theorem~\ref{thm:naive_replication_optimal}) and MDS coding (Theorem~\ref{thm:mds_optimal_prob}). We next turn to privacy guarantees.
    \subsection{Privacy}
    \label{sec:privacy}
    \proposal aims for until decode privacy before $\Kneed$ verified symbols for~$\ID$ are published, the adversary should learn essentially only the bits contained in the revealed symbols and the public header~$\Hpub$. Let the adversary's pre-decode view of a transaction be captured by the leakage function
\[
      \mathcal{L}(tx)
      := \Bigl(
          \Hpub,\
          \bigl\{(i,j,y_j)\ \text{for each verified symbol observed}\bigr\}
         \Bigr),
\]

where the set includes lane indices $i$, symbol indices $j$, and symbol values $y_j$. Bundle signatures are also visible to the adversary but depend only on $\ID$ and the observed $(j,y_j)$ values.

    \begin{lemma}[Privacy until $K$ symbols]
\label{lem:until-K-privacy}
For any $r<K$ observed verified symbols,
\[
H_\infty\!\left(\Payload\ \middle|\ \mathcal{L}(tx)\ \text{with } r\text{ symbols}\right)
\ge H_\infty(\Payload) - r\,\ell_{\mathrm{sym}}.
\]
Indices and lane labels may leak structure (e.g., sampler bias) but no additional payload bits beyond $r\,\ell_{\mathrm{sym}}$.
\end{lemma}
\begin{proof}
The adversary's view consists of the public header $\Hpub$ (including the commitment $C$), and the observed symbols $\{(j, y_j)\}$ for $r$ indices, along with the corresponding bundle signatures. By the information-theoretic properties of rateless codes, each symbol $y_j$ is a (pseudo)random linear combination of input blocks, contributing at most $\symb$ bits of information about the payload. By the hiding property of $\Com$, the commitment $C$ leaks at most negligible information about $\Payload$. The bundle signatures $\Sigma_i$ depend only on the transaction ID, the indices, the symbol values, and the sender's key; they do not reveal unobserved portions of the payload. Thus the adversary's information about the payload is bounded by the $r \cdot \symb$ bits contained in observed symbols.
\end{proof}

    We next bound the adversary's probability of early decode reconstructing the payload before any honest user could reasonably expect inclusion. Let $A$ be the number of addressed lanes that are controlled by the adversary
    (so $A\sim\mathrm{Hypergeom}(n,c_e,m)$). If each adversarial lane receives a bundle with $s$ symbols, then the adversary obtains $A\cdot s$ symbols in that slot.
    
    \begin{lemma}[Adversarial Early Decode Probability]
    \label{lem:early-decode}
    Let $A\sim\mathrm{Hypergeom}(n,c_e,m)$ be the number of adversarially controlled addressed lanes. If each bundle carries $s$ symbols, the probability that the adversary can attempt reconstruction in a slot is
    \[
      \Pr\!\big[A \ge \lceil K/s\rceil\big].
    \]
    If $m s < K$, this probability is zero. 
    \end{lemma}
    \begin{proof}
    The adversary's symbol budget is $A\cdot s$. Early decode requires $A\cdot s \ge \Kneed$, i.e.\ $A\ge \lceil \Kneed/s\rceil$, which yields the stated probability. If $m\,s < \Kneed$ then even if $A=m$ (full control of all
    addressed lanes) the adversary holds fewer than $\Kneed$ symbols, so early decode is impossible.
    \end{proof}
    
    Lemmas~\ref{lem:until-K-privacy} and~\ref{lem:early-decode} quantify the trade off between the number of addressed lanes $m$, symbols per bundle $s$, and the decoder threshold $\Kneed$, increasing $m$ and $s$ improves liveness but also increases the adversary's early decode probability. Section~\ref{sec:performance} compares these trade offs for variants of \proposal. 
    
    \section{Performance Comparison}
    \label{sec:performance}
    In this section, we compare the performance of the \proposal protocol (using verifiable rateless coding) against two variants: $\PiStr$ that uses naive replication, and $\PiMDS$ that uses fixed-rate MDS coding. The primary goal is to understand the trade-offs between these strategies, particularly concerning bandwidth consumption, under varying parameters such as payload size, metadata overhead, and the sender's desired censorship tolerance ($c_e$) and reliability ($\delta$). We refer to the variant of \proposal shown in the previous section as $\PiRTL$. 
    
    We first define the cost metrics for each approach in terms of total published bytes and the minimum bytes required for successful inclusion. We then analyze the asymptotic overhead for large payloads to understand fundamental scalability limits. While naive replication is simple, it often incurs significant overhead and sacrifices pre-finality privacy. MDS coding offers better bandwidth efficiency for large payloads but requires careful parameter selection ($m, k$) upfront and typically has zero decoding failure probability ($\delta_{code}=0$). Rateless encoding aims to provide flexibility and potentially lower overhead, especially when accounting for metadata and probabilistic guarantees. We analyse these differences, paving the way for evaluating concrete end-to-end latency and other empirical metrics in Section~\ref{sec:eva}.
    \subsection{Bandwidth Cost}
    Let $\M_h$ denote the per-bundle header overhead in bytes (the components of $Bundle_i$ excluding symbol data) and $\M_s$ the per-symbol metadata (the index $j$). Each of the $m$ lanes receives a bundle containing $s$ symbols in the rateless case, one share per lane in the MDS case (no need then to use $\M_s$), or the full transaction in the naive case.

    
    We analyze the bandwidth cost for each submission strategy, defining both the total published bytes (if all $m$ addressed lanes publish their data) and the minimum required bytes for a successful decode in a single slot in Proposition~\ref{prop:bytes_comparison}.
    \begin{proposition}[Bandwidth Cost]
\label{prop:bytes_comparison}
Let $m$ be the number of addressed lanes. Write $L_{\mathrm{pub}}$ for total published bytes and $L_{\mathrm{min}}$ for the minimum bytes required for inclusion.

\begin{enumerate}
\item[\textbf{(a)}] \textbf{Naive:} $L_{\mathrm{pub}} = m(\M_h + S)$, \quad $L_{\mathrm{min}} = \M_h + S$.

\item[\textbf{(b)}] \textbf{MDS $(m,k)$:} $L_{\mathrm{pub}} = m(\M_h + S/k)$, \quad $L_{\mathrm{min}} = k(\M_h + S/k)$.

\item[\textbf{(c)}] \textbf{Rateless:} $L_{\mathrm{pub}} = m(\M_h + s(\M_s + \symb))$, \quad $L_{\mathrm{min}} = \lceil K/s \rceil(\M_h + s(\M_s + \symb))$.
\end{enumerate}
\end{proposition}

\begin{proof}
Each lane publishes one bundle. For naive replication, inclusion requires 1 honest lane; for MDS, $k$ lanes; for rateless, $\lceil K/s \rceil$ lanes (each contributing $s$ symbols toward the $K$ needed). Multiplying per-bundle size by the relevant lane count gives both expressions.
\end{proof}
    Proposition~\ref{prop:bytes_comparison} provides exact costs including metadata, understanding the fundamental scalability requires examining the asymptotic behavior for large payloads, where metadata costs become negligible relative to the payload data itself. Corollary~\ref{cor:overhead_floor_comparison} derives the asymptotic overhead floor, defined as the ratio of total published bytes to the original payload size ($L_{pub}/S$), for each strategy.
    \begin{corollary}[Comparison of Asymptotic Overhead Floors]
\label{cor:overhead_floor_comparison}
Ignoring metadata costs (assuming $S$ is large) and analyzing the probabilistic case, the asymptotic byte overhead factor ($L_{\mathrm{pub}} / S$) for each strategy reveals its fundamental scalability:

\begin{enumerate}
    \item[\textbf{(a)}] \textbf{Naive Replication:} The overhead floor is $m_{\text{opt}}$, the smallest $m$ satisfying the probabilistic constraint $\Pr[H < 1] \le \delta$.

    \item[\textbf{(b)}] \textbf{MDS Coding:} The overhead floor is determined by the maximum reliable code rate, which is the uncensored lane ratio. The floor is therefore $\mathbf{\frac{1}{1 - c_e/n}}$. This overhead scales with the \textbf{ratio} of censors.

    \item[\textbf{(c)}] \textbf{Rateless Coding:} The overhead floor is $\mathbf{\frac{1 + \varepsilon}{1 - c_e/n}}$. This matches the scalability of MDS coding, with an explicit overhead factor of $(1+\varepsilon)$ paid for the flexibility of the rateless scheme.
\end{enumerate}
\end{corollary}
\begin{proof}
Immediate from Proposition~\ref{prop:bytes_comparison} by taking $S \to \infty$. Metadata terms vanish, leaving overhead determined by the replication/coding rate. For MDS and rateless, the effective rate is $(1-c_e/n)$; rateless pays an additional $(1+\varepsilon)$ factor for decode threshold $K$.
\end{proof}
    While Corollary~\ref{cor:overhead_floor_comparison} shows coding's asymptotic superiority for large payloads ($S \to \infty$), metadata overhead significantly impacts performance for practical transaction sizes and can make naive replication competitive:
    
    \begin{itemize}
        \item \textbf{For Naive Replication and MDS Coding,} the per-symbol metadata within a lane $\M_s$ does not apply, i.e. $M=\M_h$. Its total impact on the overhead ratio is proportional to $m\cdot \M_h/S$, which diminishes directly as $S$ grows. The additional indexing in the MDS case is negligible extra metadata compared to the rest of metadata $M_h$ required in both the naive and MDS case. However, naive replication, because of requiring only one successful lane, could lead to less lanes containing any copy/share, which could pay off for transactions with a payload of similar size to $M_h$ (i.e. $S\approx 100s$ of bytes). For medium-to-large payloads, MDS coding provides a better duplication factor.
    
        \item \textbf{For Rateless Coding,} the overhead is approximately $\left(1+\frac{\M_h+s\M_s}{s\symb}\right)$ times the asymptotic floor from Corollary~\ref{cor:overhead_floor_comparison}. Using smaller symbols (decreasing $\symb$) makes the encoding more granular but increases the relative cost of per-symbol metadata, $\M_s$. For finite $\M=\M_h +s\M_s$, one has $L_{\mathrm{pub}_\mathrm{rtl}}(\ID)/S \approx \frac{1+\epsc}{1-c_e/n}\big(1+\M/\symb\big)$. Subject to block/mempool limits, taking $\symb\gg\M$ drives the metadata term down; the intrinsic floor $\frac{1+\epsc}{1-c_e/n}$ remains. Thus, in practice, rateless coding approaches the bandwidth efficiency of MDS for payloads in the order of several KBs, where the extra metadata is amortized, while providing better end-to-end latency (Figure~\ref{fig:latency}).
    \end{itemize}
    Therefore, the optimal strategy for transactions can be heavily influenced by metadata costs. The next sections calculate the best strategy for reduced bandwidth depending on the approach, $\M_h, \M_s,\text{ and } S$, along with the failure budget and censorship-tolerance. We evaluate concrete implementations of each variant in Section~\ref{sec:eva}.
    \subsection{Optimal Submission Strategies}
    \label{sec:opt-sub}
    A sender's goal is to select the optimal submission strategy that minimizes total bandwidth cost while meeting their desired probabilistic liveness guarantee (i.e., a failure probability of at most $\delta$). The optimal configuration for each strategy is defined by the solution to the following optimization problems. Theorem~\ref{thm:naive_replication_optimal} shows that, for naive replication, the only parameter to optimize is the number of lanes $m$, aiming for at least one successful transmission.
    
    \begin{theorem}[Optimal Configuration for Naive Replication]
    \label{thm:naive_replication_optimal}
    To minimize bandwidth using naive replication, a sender must find the optimal number of lanes $m_{\text{opt}}$. This is the smallest integer $m$ that satisfies the probabilistic constraint:
    \[
        \Pr[H < 1] \le \delta, \quad \text{where } H \sim \mathrm{Hypergeom}(n, n-c_e, m).
    \]
    A conservative closed-form approximation for $m_{\text{opt}}$ is given by:
    \[
        \boxed{ m_{\text{opt}} \ge \left(\frac{b+\sqrt{b^2+4c_r}}{2c_r}\right)^{\!2} }
    \]
    where $c_r=1-c_e/n$ is the assumed honest validator ratio and $b = \sqrt{2c_r\ln(1/\delta)}$.
    \end{theorem}
    
    \begin{proof}
    Cost $m(M_h + S)$ is strictly increasing in $m$, so the optimum is the smallest $m$ satisfying the constraint. The closed-form follows from a Chernoff bound on the hypergeometric tail.
    \end{proof}
    MDS coding introduces the parameter $k$ (number of shares needed for reconstruction). Optimizing involves finding the best pair $(m, k)$ that minimizes cost while ensuring at least $k$ honest lanes succeed. This typically requires a numerical search over possible values of $k$, determining the minimum required $m$ for each $k$ using the constraint, as we show in Theorem~\ref{thm:mds_optimal_prob}
    \begin{theorem}[Optimal Configuration for MDS Coding]
    \label{thm:mds_optimal_prob}
    To minimize total bandwidth using an $(m, k)$ MDS code, a sender must find the optimal integer pair $(m_{\text{opt}}, k_{\text{opt}})$ that solves the optimization problem:
    \[
        \min_{m,k} \left[ m \cdot \left(\M_h + \frac{S}{k}\right) \right]
    \]
    For any choice of $k$, the number of lanes $m$ must satisfy the probabilistic constraint:
    \[
        \Pr[H < k] \le \delta, \quad \text{where } H \sim \mathrm{Hypergeom}(n, n-c_e, m).
    \]
    A conservative closed-form approximation for the required $m$ given a choice of $k$ is:
    \[
        \boxed{ m \ge \left(\frac{b+\sqrt{b^2+4c_r k}}{2c_r}\right)^{\!2} }
    \]
    where $c_r=1-c_e/n$ and $b = \sqrt{2c_r\ln(1/\delta)}$.
    \end{theorem}
    \begin{proof}
Analogous to Theorem~\ref{thm:naive_replication_optimal}: for each $k$, find the smallest $m$ satisfying the constraint, then minimize over $k$.
\end{proof}
    Rateless coding offers more flexibility by optimizing the number of lanes $m$, symbols per bundle $s$, and total symbols needed $K$. The goal is to find the tuple $(m, s, K)$ minimizing cost while ensuring at least $\lceil K/s \rceil$ honest lanes succeed. This involves a multi-variable search over $s$ and $K$, determining the necessary $m$ for each pair via the constraint. Note the constraint adjusts the target probability to account for the intrinsic decoding failure probability $\delta_{code}$ of the rateless scheme itself. We show this in Theorem~\ref{thm:rateless_optimal_prob}.
    \begin{theorem}[Optimal Configuration for Rateless Coding]
\label{thm:rateless_optimal_prob}
To minimize bandwidth using the rateless scheme, a sender must find the optimal tuple $(m_{\text{opt}}, s_{\text{opt}}, {\symb}_{\text{opt}})$ that solves the optimization problem:
\[
    \min_{m,s,\symb} \left[ m \cdot \left(\M_h + s \cdot (\M_s + \symb)\right) \right]
\]
where the decode threshold is determined by the symbol size:
\[
  \Kneed(\symb) \;:=\; \left\lceil (1+\epsc)\,\frac{S}{\symb}\right\rceil.
\]
For any choice of $(s, \symb)$, the number of lanes $m$ must satisfy the probabilistic constraint:
\[
    \Pr\bigl[H < \lceil\Kneed(\symb)/s\rceil\bigr] \le \delta - \delc, \quad \text{where } H \sim \mathrm{Hypergeom}(n, n-c_e, m).
\]
A conservative closed-form approximation for the required $m$ is given by:
\[
    \boxed{ m \ge \left(\frac{b+\sqrt{\,b^2+4c_r\,\Kneed'(\symb)\,}}{2c_r}\right)^{\!2} }
\]
where $\Kneed'(\symb) = \lceil\Kneed(\symb)/s\rceil$ is the required number of honest lanes, $c_r=1-c_e/n$, and $b = \sqrt{2c_r\ln(1/(\delta-\delc))}$.
\end{theorem}

\begin{proof}
The optimization is a multi-variable search over the symbols per bundle $s$ and the symbol size $\symb$. The decode threshold $\Kneed(\symb)$ is not an independent variable but is determined by the choice of $\symb$ via the rateless code parameterization. For each candidate pair $(s, \symb)$, the minimum required $m$ is the smallest integer satisfying the probabilistic constraint, for which the provided closed-form bound serves as a conservative estimate. The optimal tuple is the one yielding the minimum total cost across this search space.
\end{proof}
    
    These theorems provide the mathematical basis for determining the most bandwidth-efficient configuration for each submission strategy under given specific system's conditions $(n)$ and  user requirements ($c_e, \delta, S$). Section~\ref{sec:eva} provides concrete comparisons in practice, depending on each of the relevant user parameters.
    
    \subsection{Comparison to the Deterministic Lower Bound}
    To contextualize the performance of these probabilistic strategies, we compare them to the fundamental limit of any deterministic coding scheme.
    
    \begin{theorem}[Information-theoretic lower bound, deterministic threshold]
    \label{thm:it_lower_bound}
    Assume decoding must succeed from any set of $n-c_e$ lanes and fail from any smaller set. Let lane $i$ contribute a total length $\ell_i$ (possibly via multiple shares), and set $L:=\sum_{i=1}^n \ell_i$. Then
    \[
    L \ \ge\ \frac{n}{n-c_e}\,S,
    \]
    so the overhead factor $L/S$ is at least $1/(1-c_e/n)$.
    \end{theorem}
    \begin{proof}
    Let $k:=n-c_e$. For every subset $J$ of lanes of size $k$, the total contributed length satisfies $\sum_{j\in J}\ell_j\ge S$. Summing over all $\binom{n}{k}$ such subsets and noting that each $\ell_i$ appears exactly $\binom{n-1}{k-1}$ times yields $\binom{n-1}{k-1}L\ge\binom{n}{k}S$. Using $\binom{n}{k}=\frac{n}{k}\binom{n-1}{k-1}$ gives $L\ge \frac{n}{k}S=\frac{n}{n-c_e}S$.
    \end{proof}
    The information-theoretic lower bound derived in Theorem~\ref{thm:it_lower_bound} for any deterministic scheme ($1/(1-c_e/n)$) matches the asymptotic overhead floor previously calculated for probabilistic MDS and (up to $1+\epsilon$) Rateless coding strategies (Corollary~\ref{cor:overhead_floor_comparison}). This demonstrates that these probabilistic approaches are asymptotically optimal in terms of bandwidth overhead.
    \section{Evaluation}
\label{sec:eva}

We evaluate \proposal through a combination of analytical computation and microbenchmarks. Our evaluation addresses three primary questions: (1) How do the bandwidth overheads of naive replication, MDS coding, and rateless coding compare across different payload sizes and censorship assumptions? (2) What is the computational cost of encoding and decoding, and does it introduce prohibitive latency? (3) How do the three approaches differ in their privacy properties, specifically the adversary's probability of early payload reconstruction?

\subsection{Experimental Setup}

Unless otherwise specified, we use the following default parameters throughout our evaluation. We consider a validator set of $n = 256$ validators with an effective censorship ratio of $c_e/n = 0.125$. The target per-slot failure probability is $\delta = 10^{-9}$. 

 For the rateless scheme, we set the coding overhead parameter $\varepsilon = 0.05$, per-bundle header metadata $M_h = 200$ bytes (e.g., 32 bytes for the transaction identifier, 32 bytes for the commitment, 64 bytes for the header signature, 64 bytes for the bundle signature, and a small amount for lane index and other header fields), per-symbol metadata $M_s = 8$ bytes (symbol index), and symbol size $\ell_{\mathrm{sym}} = 256$ bytes.


Individual figures may use different parameter values to illustrate specific phenomena; deviations from the defaults are noted in figure captions.

\subsection{Bandwidth Overhead}

We begin by examining the core bandwidth efficiency claims of \proposal. Figure~\ref{fig:overhead_payload} presents the overhead factor $L_{\mathrm{pub}}/S$ as a function of payload size for all three submission strategies, where $L_{\mathrm{pub}}$ denotes total published bytes and $S$ is the original payload size.

\begin{figure}[t]
    \centering
    \includegraphics[width=\columnwidth]{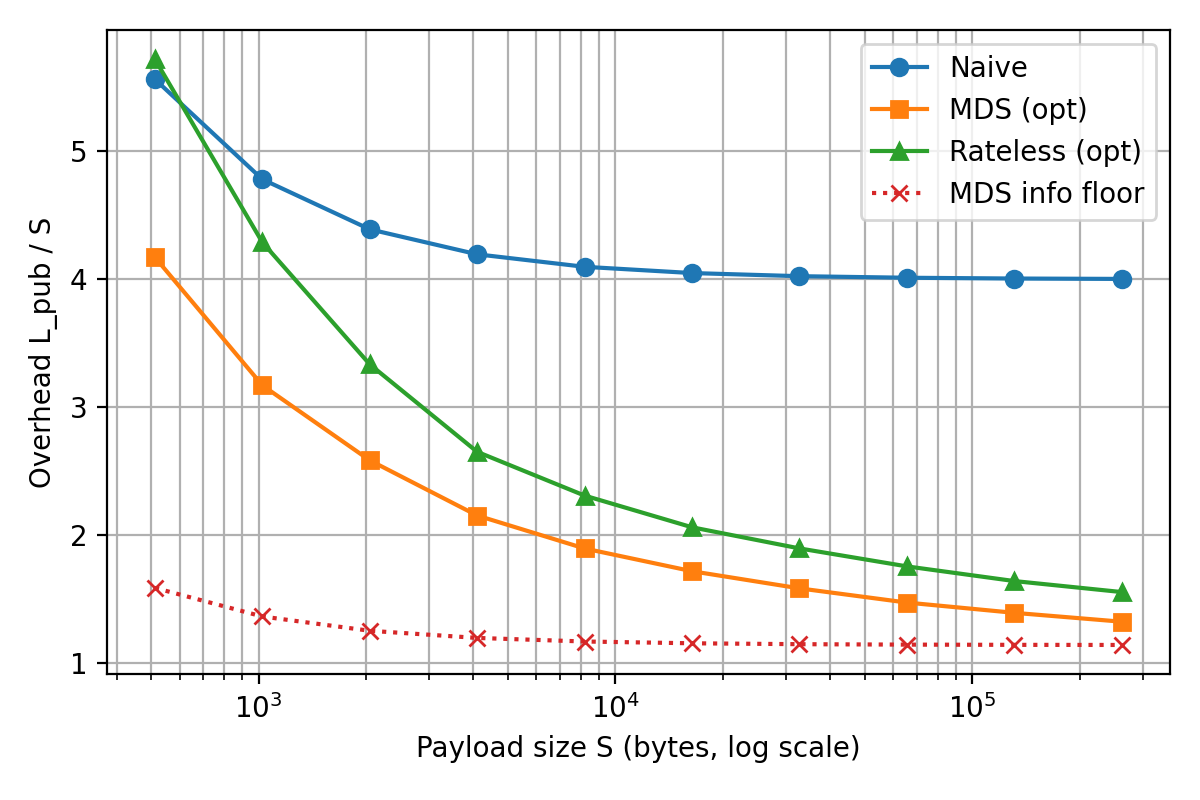}
    \caption{Bandwidth overhead factor as a function of payload size. The dotted line shows the information-theoretic lower bound $1/(1-c_e/n)$ from Theorem~\ref{thm:it_lower_bound}, for $n = 256$, $c_e/n = 0.125$, $\delta = 10^{-9}$.}    \label{fig:overhead_payload}
\end{figure}

Several observations emerge from this figure. First, naive replication maintains a nearly constant overhead across all payload sizes, reflecting the fact that the replication factor $m$ is determined solely by the censorship tolerance requirement and is independent of $S$. Second, both MDS and rateless coding exhibit decreasing overhead as payload size increases, asymptotically approaching the information-theoretic floor of $1/(1-c_e/n)$ established in Theorem~\ref{thm:it_lower_bound}. Third, for small payloads, metadata overhead causes all three approaches to converge, and naive replication can be competitive due to its lower per-bundle metadata requirements.

The rateless scheme incurs slightly higher overhead than MDS due to the $(1+\varepsilon)$ factor and additional per-symbol metadata $M_s$. However, this gap narrows for larger payloads as the metadata cost is amortized.

\subsection{Single-Slot Success Probability}

The liveness guarantees of each approach depend critically on the relationship between the number of addressed lanes $m$ and the probability of successful inclusion in a single slot. Figure~\ref{fig:success_comparison} compares the single-slot success probability as a function of fanout $m$ for all three strategies.

\begin{figure}[t]
    \centering
    \includegraphics[width=\columnwidth]{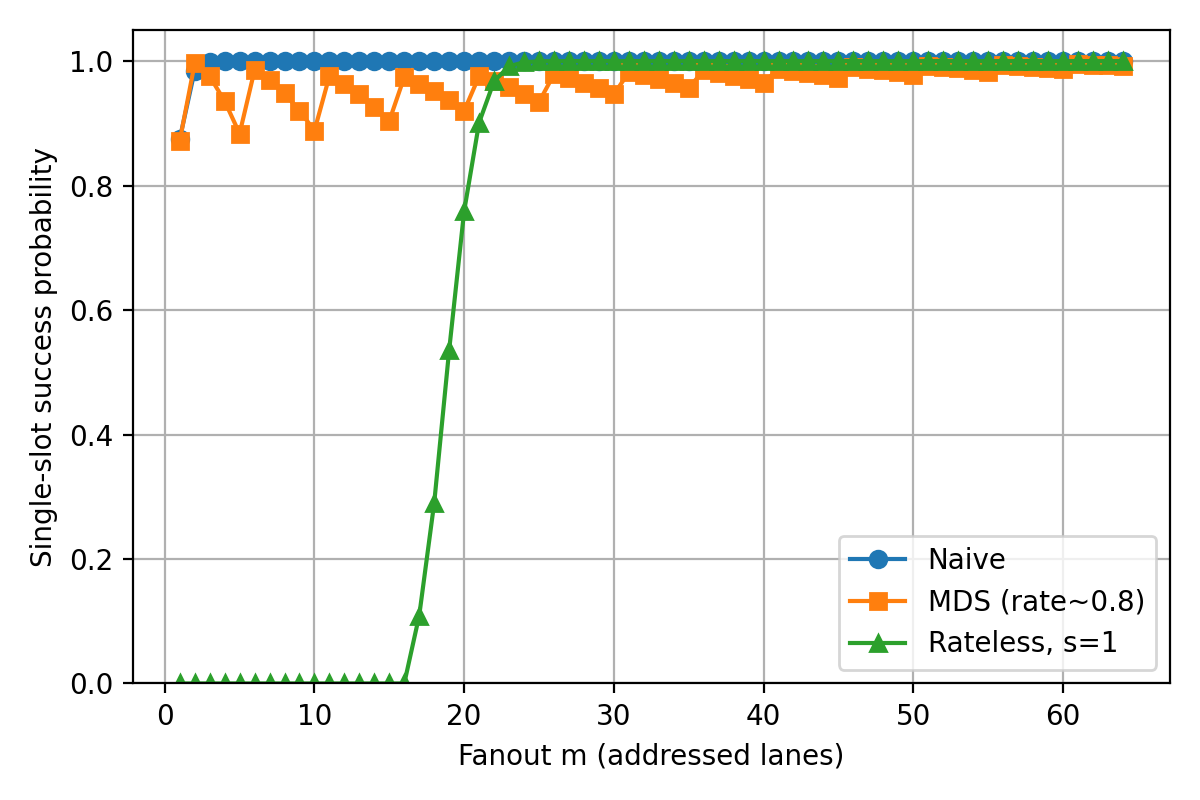}
    \caption{Single-slot success probability versus fanout $m$ for naive replication, MDS coding, and rateless coding. Parameters: $n = 256$, $c_e/n = 0.125$, $\delta = 10^{-9}$, $S = 4096$ bytes.}    \label{fig:success_comparison}
\end{figure}

Naive replication achieves high success probability at low fanout because inclusion requires only a single honest lane to publish the transaction. In contrast, both coded approaches exhibit threshold behavior: success probability remains near zero until the fanout reaches a critical value, then rises sharply to near certainty. This threshold corresponds to the point at which the expected number of honest lanes exceeds the decoding requirement ($k$ for MDS, $\lceil K/s \rceil$ for rateless).

Users can shift the rateless threshold by adjusting the symbols-per-lane parameter $s$, as illustrated in Figure~\ref{fig:s_effect}.

\begin{figure}[h]
    \centering
    \includegraphics[width=\columnwidth]{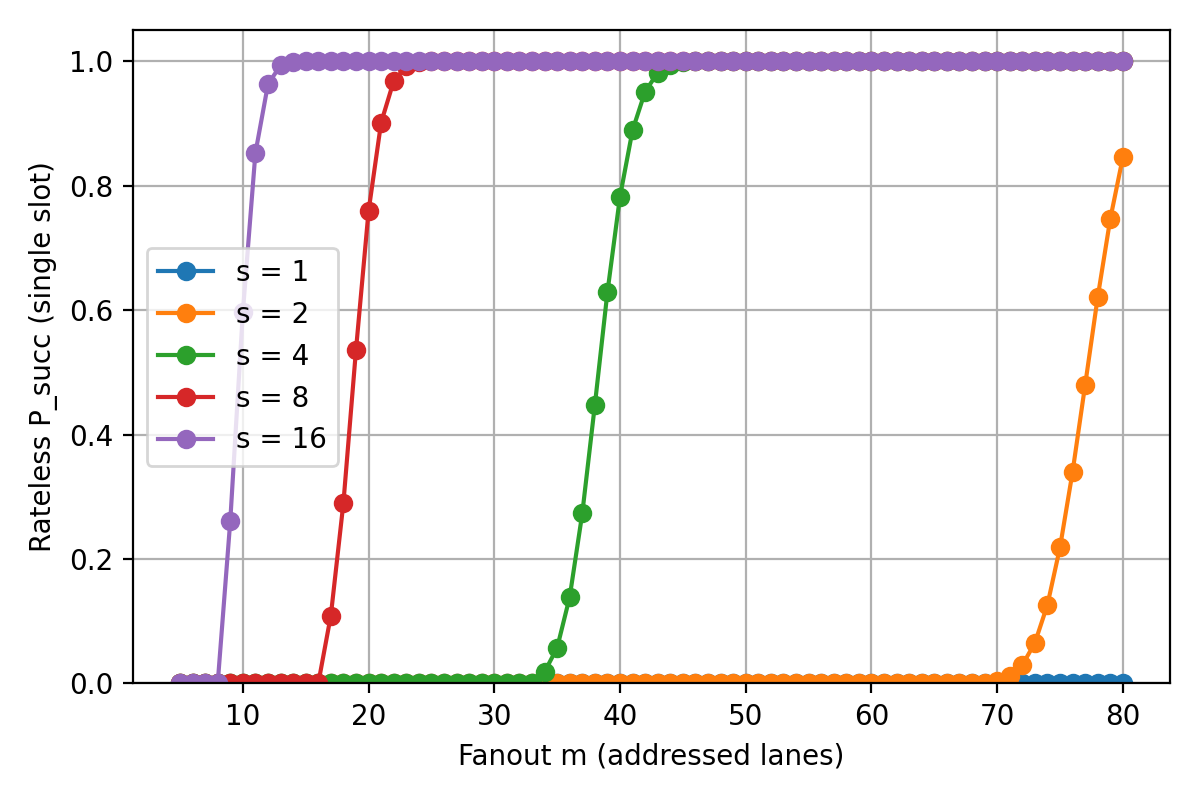}
    \caption{Effect of symbols per lane $s$ on rateless success probability. Increasing $s$ reduces the required number of honest lanes at the cost of larger bundle sizes. Parameters: $n = 256$, $c_e/n = 0.125$, $\delta = 10^{-9}$, $S = 32768$ bytes.}
    \label{fig:s_effect}
\end{figure}

Higher values of $s$ shift the success curve leftward, enabling users to achieve target success probabilities with fewer contacted lanes. This flexibility allows users to navigate the tradeoff between the number of network connections (which affects dissemination latency) and bundle size (which affects per-lane bandwidth). The optimal choice depends on application-specific constraints such as network topology and validator bandwidth limits.

\FloatBarrier

\subsection{Impact of Censorship Assumptions}

The bandwidth overhead of coded approaches depends fundamentally on the assumed censorship ratio $c_e/n$. Figure~\ref{fig:overhead_censorship} illustrates how overhead scales with censorship tolerance for rateless coding.

\begin{figure}[t]
    \centering
    \includegraphics[width=\columnwidth]{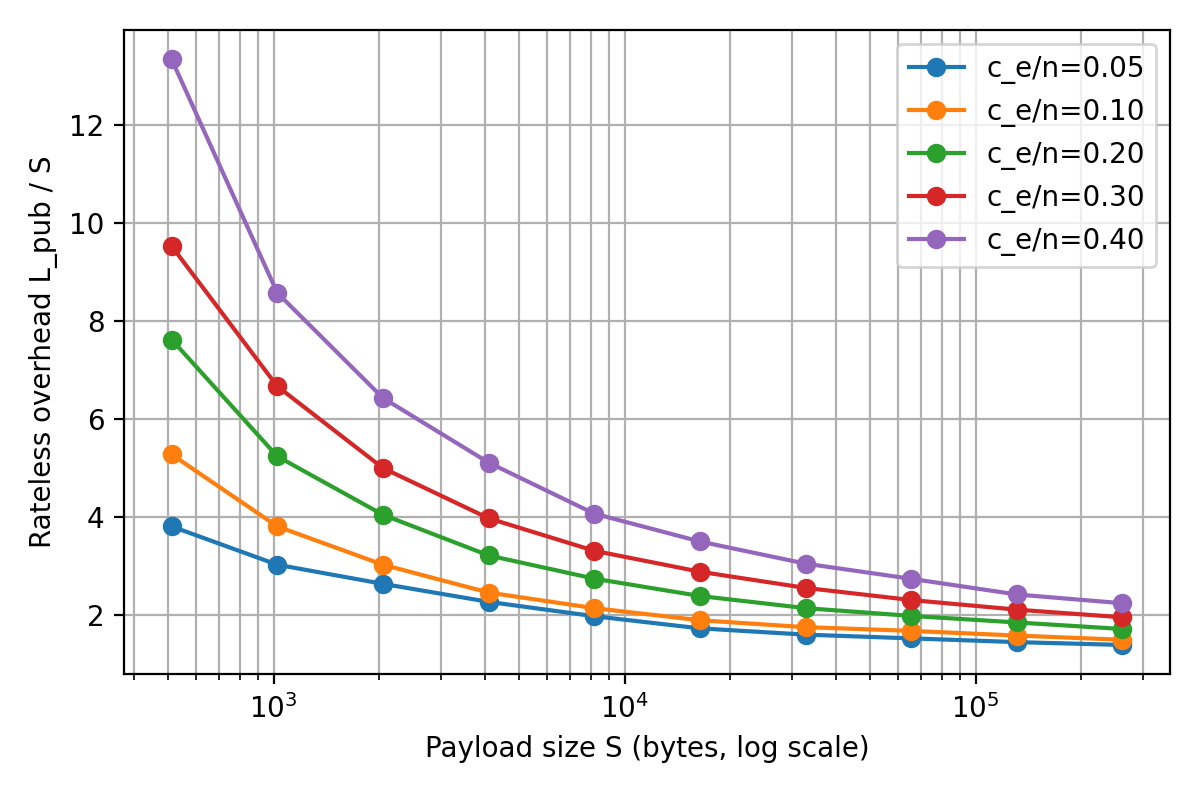}
    \caption{Rateless coding overhead as a function of payload size for varying censorship ratios $c_e/n \in \{0.05, 0.10, 0.20, 0.30, 0.40\}$. Higher censorship tolerance requires proportionally more redundancy, but overhead decreases with payload size in all cases. Parameters: $n = 256$, $\delta = 10^{-9}$.}
    \label{fig:overhead_censorship}
\end{figure}

As expected from Corollary~\ref{cor:overhead_floor_comparison}, the asymptotic overhead floor increases with censorship tolerance according to $(1+\varepsilon)/(1-c_e/n)$. Importantly, even at high censorship ratios, coded approaches substantially outperform naive replication for large payloads, as shown in Figure~\ref{fig:overhead_payload}.
\FloatBarrier
\subsection{Reliability-Cost Tradeoff}

Users can select their desired failure probability $\delta$ based on application requirements. Figure~\ref{fig:overhead_delta} shows how bandwidth overhead varies with the failure budget for each submission strategy.

\begin{figure}[t]
    \centering
    \includegraphics[width=\columnwidth]{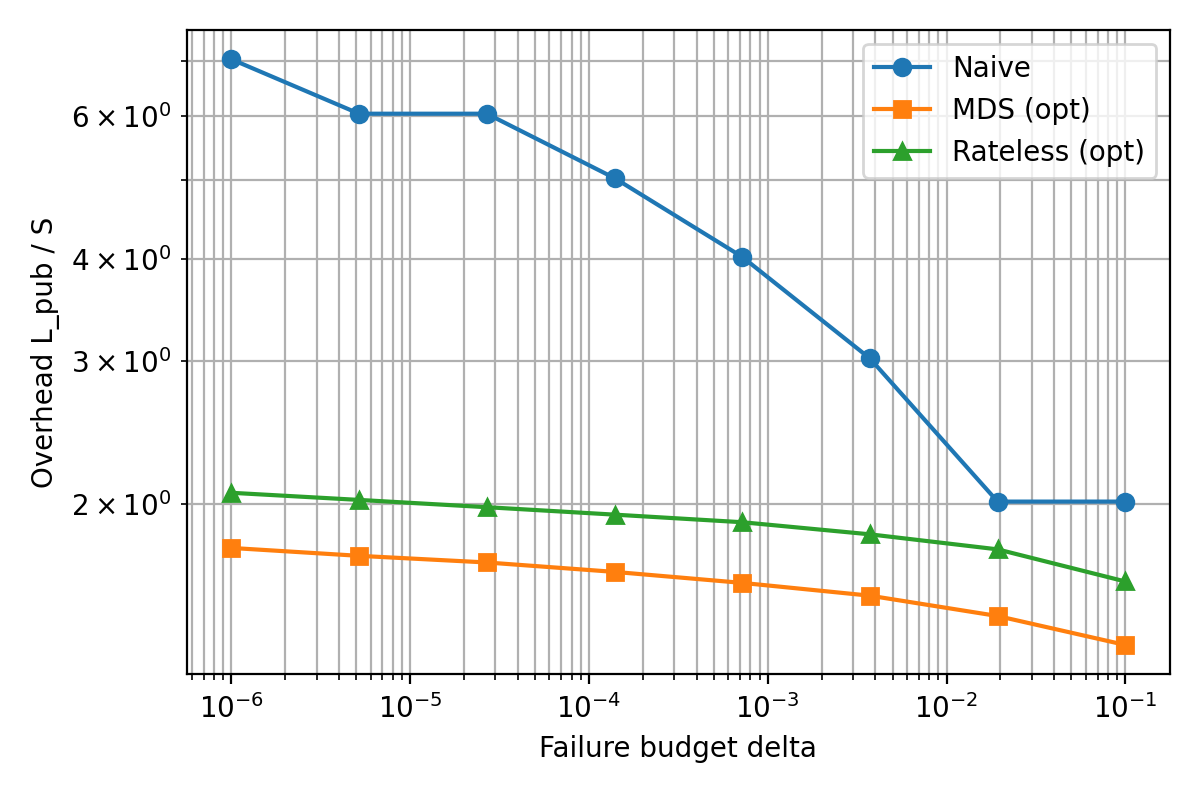}
    \caption{Bandwidth overhead as a function of failure budget $\delta$. Lower failure probability (higher reliability) requires additional redundancy. Coded approaches scale more gracefully than naive replication as reliability requirements tighten. Parameters: $n = 256$, $c_e/n = 0.125$, $S = 4096$ bytes.}
    \label{fig:overhead_delta}
\end{figure}

All approaches require increased resources as the failure budget decreases, but naive replication exhibits the steepest growth. This difference arises because naive replication must add entire copies of the payload to improve reliability, whereas coded approaches can add smaller increments of redundancy. For applications requiring high reliability, the bandwidth advantage of coding becomes more pronounced.
\FloatBarrier
\subsection{Privacy and Early Decode Probability}

A key advantage of \proposal over naive replication is the reduction in pre-inclusion information leakage. With naive replication, any addressed lane that is adversarially controlled immediately learns the full transaction payload. With coded approaches, the adversary must collect at least $k$ (MDS) or $K$ (rateless) symbols before reconstruction is possible. Figure~\ref{fig:early_decode} quantifies the adversary's probability of early decode.

\begin{figure}[t]
    \centering
    \includegraphics[width=\columnwidth]{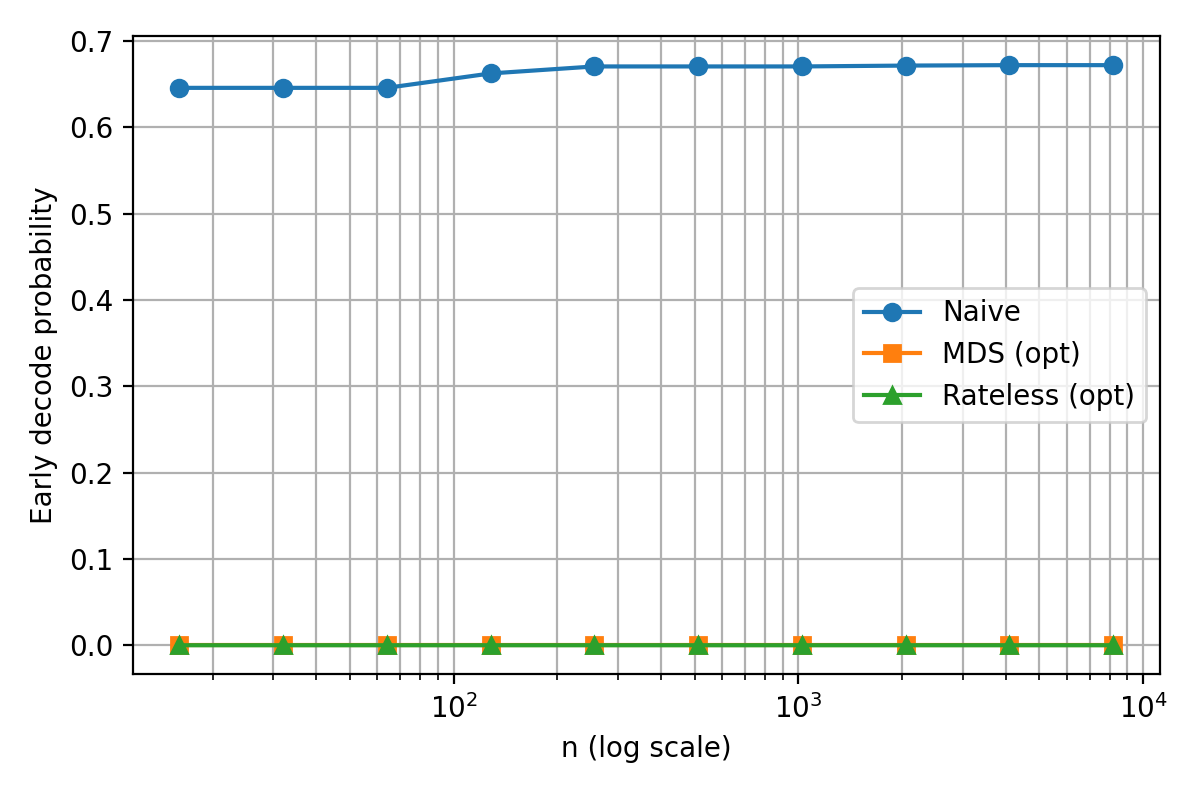}
    \caption{Early decode probability as a function of validator set size $n \in [16, 8192]$. Parameters: $c_e/n = 0.2$, $\delta = 10^{-9}$.}
    \label{fig:early_decode}
\end{figure}

Under naive replication, the adversary's early decode probability is trivial, reflecting the probability that at least one of the $m$ addressed lanes is adversarially controlled. In contrast, both MDS and rateless coding maintain early decode probability near zero across tested configurations, as shown in Figure~\ref{fig:early_decode}. This occurs because the adversary would need to control at least $k$ or $\lceil K/s \rceil$ of the addressed lanes, which is unlikely when parameters are chosen appropriately.

This privacy improvement directly translates to reduced MEV exposure. Under naive replication, a user's transaction is visible to the adversary with high probability before any honest proposer includes it, enabling front-running, sandwiching, and other extraction strategies. Coded approaches ensure that the payload remains hidden until honest validators collectively publish sufficient symbols for decoding, at which point inclusion is imminent.

\subsection{Computational Overhead}

While the primary focus of \proposal is bandwidth efficiency, computational overhead must also be practical for real-world deployment. Figure~\ref{fig:latency} presents latency measurements for encoding and decoding operations.

\begin{figure}[t]
    \centering
    \includegraphics[width=\columnwidth]{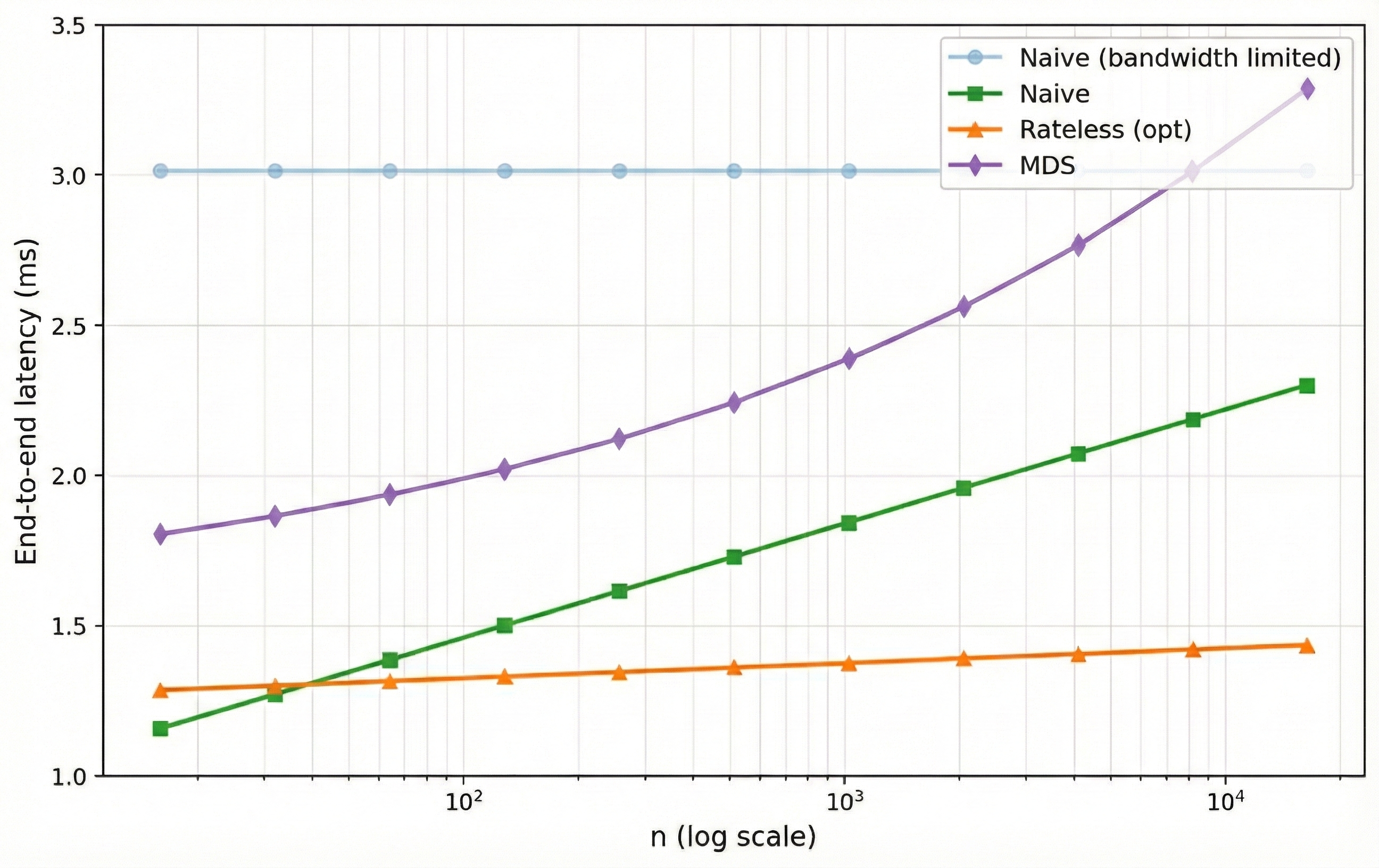}
    \caption{End-to-end latency as a function of validator set size $n \in [16, 8192]$ for naive replication (with and without bandwidth limits), MDS coding, and rateless coding. Rateless coding achieves lower latency than naive replication despite additional encoding overhead, as bandwidth savings from reduced data transmission dominate computational costs. Parameters: $c_e/n = 0.2$, $\delta = 10^{-9}$.}
    \label{fig:latency}
\end{figure}

Naive replication serves as a baseline with effectively zero encoding overhead, requiring only signature generation and serialization. Rateless coding introduces additional computational work for symbol and signature generation and verification. Interestingly, the rateless scheme shows lower end-to-end latency than naive replication in our measurements. This counterintuitive result arises because the bandwidth savings from coding dominate the additional computational overhead of symbol generation, bundle construction, and encoding. The reduced data volume leads to faster network transmission, more than compensating for encoding costs. MDS coding, while providing better bandwidth utilization, degrades in latency with the number of proposers, despite MDS being SIMD-accelerated in our prototype, justifying our choice for rateless encoding as the recommended, scalable protocol version for \proposal.

\subsection{Case Study: Real World Workloads}

To ground our analysis in real-world systems, we consider transaction dissemination patterns observed in Filecoin, an MCP blockchain where messages are broadcast via gossip and proposers independently select from their mempools. Empirical observations indicate average replication factors of 3--4$\times$, with some transactions replicated up to 9$\times$ due to uncoordinated submission and deduplication failures.

\begin{figure}[t]
    \centering
    \includegraphics[width=\columnwidth]{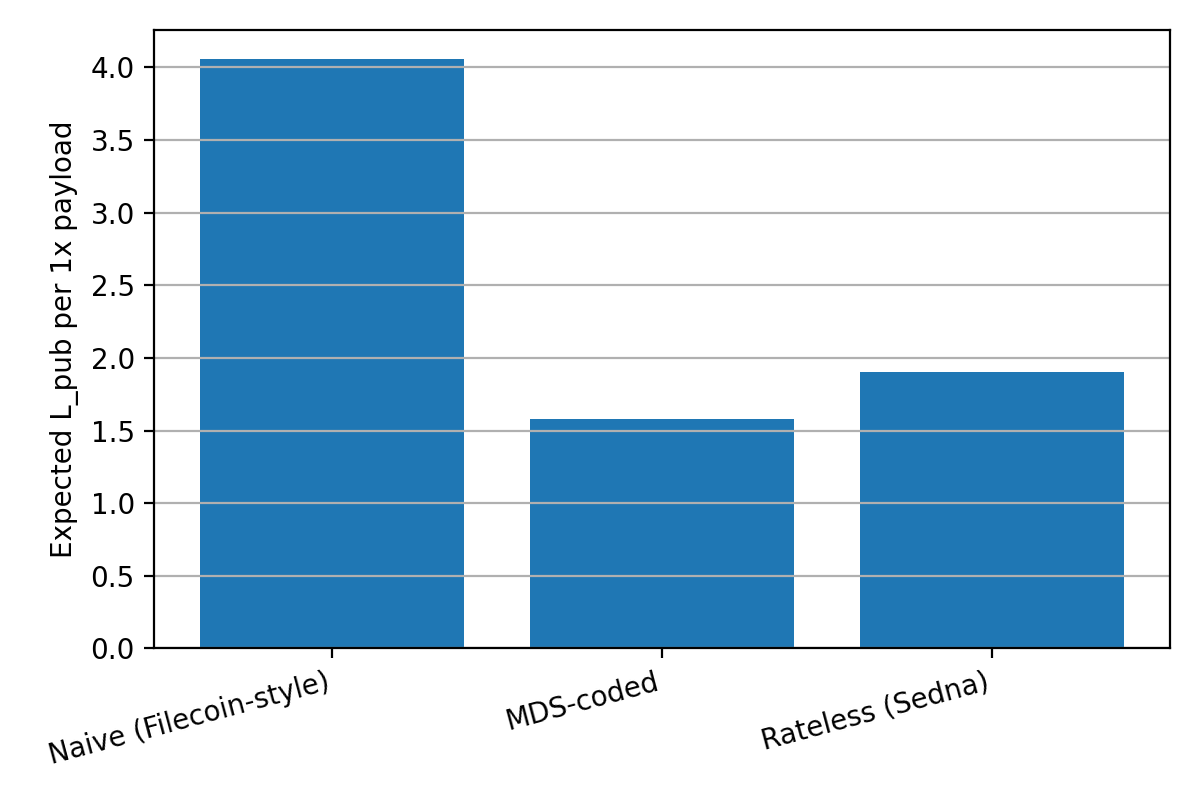}
    \caption{Comparison of replication factors with naive replication and Sedna's MDS and rateless variants, for the same level of censorship resistance. Parameters: $n = 256$, $c_e/n = 0.125$, $\delta = 10^{-9}$.}
    \label{fig:filecoin}
\end{figure}

For systems that charge users for all duplicates, the savings from coded approaches accrue directly to users. Under models where the system absorbs duplication costs, \proposal would improve network goodput for equivalent censorship resistance guarantees.

\subsection{Summary}

Our evaluation demonstrates that \proposal achieves its design goals across a range of parameters. For medium-to-large payloads, rateless coding reduces bandwidth overhead compared to naive replication while providing strong privacy guarantees against early payload reconstruction. MDS coding achieves slightly lower overhead but requires upfront commitment to code parameters. Computational costs remain practical for tested payload sizes. The flexibility of the rateless scheme allows users to navigate the trilemma of censorship resistance, latency, and cost according to their application-specific requirements.
    
    \section{Conclusion}
    \label{sec:conclusion}
    In this work we  presented Sedna, a user facing transaction dissemination protocol for MCP blockchains that addresses the fundamental trilemma between censorship resistance, low latency, and reasonable pricing. By replacing whole transaction replication with commitment bound, addressed bundles of verifiable coded symbols, Sedna enables users to navigate this trilemma according to their individual requirements for each transaction. Our analysis establishes several key results. First, we proved that Sedna achieves strong safety guarantees, including header/commitment non-malleability, deterministic resolution of symbol indices, and deterministic execution ordering by inclusion height, all under standard cryptographic assumptions. Second, we derived closed form bounds for per-slot inclusion probability under an assumed effective censor mass, giving users concrete prescriptions for lane selection and symbol allocation to meet target reliability levels. Third, we formalized until-decode privacy, bounding pre-inclusion information leakage to the bits contained in revealed symbols and demonstrating how this reduces the MEV surface compared to transparent mempool dissemination. The bandwidth analysis reveals that Sedna's rateless variant achieves an asymptotic overhead floor of $\frac{1+\varepsilon}{1-c_e/n}$, matching the information-theoretic lower bound for deterministic schemes up to the small rateless coding factor $\varepsilon$. We identified the parameter regimes where each variant proves most bandwidth-efficient, with naive replication remaining competitive for small payloads where metadata dominates and coded approaches becoming advantageous as payload size grows. Sedna is entirely user-facing and requires no protocol-level changes to the underlying MCP consensus, allowing heterogeneous submission strategies to coexist. This positions it as a practical layer that MCP deployments can adopt incrementally. Future work includes empirical evaluation across production MCP systems, integration with existing fee markets, and exploration of adaptive strategies that adjust encoding parameters based on observed network conditions and censor behavior.

    \section*{Acknowledgments}
    We would like to thank Ittai Abraham and Mahimna Kelkar for their useful input and conversations regarding this work.

    \bibliographystyle{abbrv}
    \bibliography{references}

@inproceedings{miller2016honeybadger,
  title        = {The Honey Badger of {BFT} Protocols},
  author       = {Andrew Miller and Yu Xia and Kyle Croman and Elaine Shi and Dawn Song},
  booktitle    = {ACM CCS},
  year         = {2016},
  url          = {https://eprint.iacr.org/2016/199.pdf}
}

@inproceedings{keidar2021dagrider,
  title = {All You Need is DAG},
  author = {Keidar, Idit and Kokoris-Kogias, Eleftherios and Naor, Oded and Spiegelman, Alexander},
  booktitle    = {PODC},
  year         = {2021},
  pages = {165–175},
}

@inproceedings{danezis2022narwhal,
    author = {Danezis, George and Kokoris-Kogias, Lefteris and Sonnino, Alberto and Spiegelman, Alexander},
    title = {Narwhal and Tusk: a DAG-based mempool and efficient BFT consensus},
    year = {2022},
    booktitle = {EuroSys},
    pages = {34–50}
}

@inproceedings{bormet2025beat,
  title={$\{$BEAT-MEV$\}$: Epochless Approach to Batched Threshold Encryption for $\{$MEV$\}$ Prevention},
  author={Bormet, Jan and Faust, Sebastian and Othman, Hussien and Qu, Ziyan},
  booktitle={34th USENIX Security Symposium (USENIX Security 25)},
  pages={3457--3476},
  year={2025}
}

@inproceedings{choudhuri2024mempool,
  title={Mempool privacy via batched threshold encryption: Attacks and defenses},
  author={Choudhuri, Arka Rai and Garg, Sanjam and Piet, Julien and Policharla, Guru-Vamsi},
  booktitle={33rd USENIX Security Symposium (USENIX Security 24)},
  pages={3513--3529},
  year={2024}
}

@article{rivaseahorse,
  title={Seahorse: Efficiently Mixing Encrypted and Normal Transactions},
  author={Riva, Ben and Sonnino, Alberto and Kokoris-Kogias, Lefteris},
    
}

@inproceedings{kokoriskogias2022bullshark,
  title = {Bullshark: DAG BFT Protocols Made Practical},
  author = {Spiegelman, Alexander and Giridharan, Neil and Sonnino, Alberto and Kokoris-Kogias, Lefteris},
  booktitle = {CCS},
  year = {2022},
  pages = {2705–2718},
}

@inproceedings{gagol2019aleph,
    author = {G\k{a}gol, Adam and Le\'{s}niak, Damian and Straszak, Damian and Swietek, Micha\l{}},
    title = {Aleph: Efficient Atomic Broadcast in Asynchronous Networks with Byzantine Nodes},
    year = {2019},
    booktitle = {AFT},
    pages = {214–228},
}

@misc{baird2016hashgraph,
  title        = {The Swirlds Hashgraph Consensus Algorithm: Fair, Fast, Byzantine Fault Tolerance},
  author       = {Leemon Baird},
  year         = {2016},
  howpublished = {Whitepaper},
  url          = {https://www.hedera.com/hh_whitepaper_v1.0-180313.pdf}
}

@inproceedings{daian2020flashboys,
    title="Flash Boys 2.0: Frontrunning in Decentralized Exchanges, Miner Extractable Value, and Consensus Instability",
    author="Daian, Philip
        and Goldfeder, Steven
        and Kell, Tyler
        and Li, Yunqi
        and Zhao, Xueyuan
        and Bentov, Iddo
        and Breidenbach, Lorenz
        and Juels, Ari",
    booktitle="IEEE S\&P",
    year="2020",
    pages="585--602"
}

@inproceedings{kelkar2020aequitas,
    title="Order-Fairness for {B}yzantine Consensus",
    author="Kelkar, Mahimna
        and Zhang, Fan
        and Goldfeder, Steven
        and Juels, Ari",
    booktitle="CRYPTO",
    year="2020",
    pages="451--480"
}

@inproceedings{kelkar2023themis,
  author = {Mahimna Kelkar and Soubhik Deb and Sishan Long and Ari Juels and Sreeram Kannan},
  title = {Themis: Fast, Strong Order-Fairness in Byzantine Consensus},
  booktitle = {CCS},
  pages = {475--489},
  year = {2023},
}

@inproceedings{mu2024speedyfair,
  title = {Separation is Good: A Faster Order-Fairness Byzantine Consensus},
  author = {Ke Mu and Bo Yin and Alia Asheralieva and Xuetao Wei},
  booktitle = {NDSS},
  year = {2024}
}

@misc{flashbotsMevBoost,
  title        = {MEV-Boost: Open Source Relay Middleware for Proposer-Builder Separation},
  author       = {{Flashbots}},
  year         = {2022},
  howpublished = {\url{https://docs.flashbots.net/}},
  note         = {Accessed 2025-11}
}

@misc{ethereumPBS,
  title        = {Proposer-Builder Separation (PBS) — Ethereum Research \& specs},
  author       = {{Ethereum Foundation}},
  year         = {2023},
  howpublished = {\url{https://ethresear.ch/t/proposer-builder-separation/}},
  note         = {Accessed 2025-11}
}

@misc{flashbotsProtect,
  title        = {Flashbots Protect RPC},
  author       = {{Flashbots}},
  year         = {2022},
  howpublished = {\url{https://docs.flashbots.net/flashbots-protect/}},
  note         = {Accessed 2025-11}
}

@misc{bloxroutePrivateTx,
  title        = {bloXroute Private Transactions},
  author       = {{bloXroute Labs}},
  year         = {2023},
  howpublished = {\url{https://docs.bloxroute.com/introduction/private-transactions/}},
  note         = {Accessed 2025-11}
}

@misc{mevBlocker,
  title        = {MEV Blocker: Private Order Flow Protection},
  author       = {{CoW Protocol et al.}},
  year         = {2023},
  howpublished = {\url{https://docs.mevblocker.io/}},
  note         = {Accessed 2025-11}
}

@misc{shutter2021whitepaper,
  title        = {Shutter: Threshold Encryption for Preventing MEV},
  author       = {{Shutter Network}},
  year         = {2021},
  howpublished = {\url{https://docs.shutter.network/}},
  note         = {Accessed 2025-11}
}

@article{budish2015fba,
  title        = {The High-Frequency Trading Arms Race: Frequent Batch Auctions as a Market Design Response},
  author       = {Eric Budish and Peter Cramton and John Shim},
  journal      = {QJE},
  year         = {2015},
  volume       = {130},
  number       = {4}
}

@misc{cow2021,
  title        = {CoW Protocol: Batch Auctions for DEX Trading},
  author       = {{CoW Protocol}},
  year         = {2021},
  howpublished = {\url{https://docs.cow.fi/}},
  note         = {Accessed 2025-11}
}

@inproceedings{luby2002lt,
  title        = {{LT} Codes},
  author       = {Michael Luby},
  booktitle    = {FOCS},
  year         = {2002},
  pages = {271},
}

@article{shokrollahi2006raptor,
  title        = {Raptor Codes},
  author       = {Amin Shokrollahi},
  journal      = {IEEE Trans. Info. Theory},
  year         = {2006},
  volume       = {52},
  number       = {6}
}

@misc{rfc6330,
  title        = {RFC 6330: RaptorQ Forward Error Correction Scheme for Object Delivery},
  author       = {Michael Luby and Amin Shokrollahi and Mark Watson and others},
  howpublished = {IETF RFC},
  year         = {2011},
  url          = {https://www.rfc-editor.org/rfc/rfc6330}
}

@misc{rfc5053,
  title        = {RFC 5053: Raptor Forward Error Correction Scheme},
  author       = {J. Peterson and others},
  howpublished = {IETF RFC},
  year         = {2007},
  url          = {https://www.rfc-editor.org/rfc/rfc5053}
}

@inproceedings{reed1960polynomial,
  title        = {Polynomial Codes over Certain Finite Fields},
  author       = {Irving S. Reed and Gustave Solomon},
  booktitle    = {J. Soc. Ind. Appl. Math.},
  year         = {1960}
}

@book{wicker1995rs,
  title        = {Reed–Solomon Codes and Their Applications},
  author       = {Stephen B. Wicker and Vijay K. Bhargava},
  publisher    = {IEEE Press},
  year         = {1995}
}

@misc{sui2022whitepaper,
  title        = {Sui: A Platform for High-Performance Smart Contracts},
  author       = {{Mysten Labs}},
  year         = {2022},
  howpublished = {\url{https://github.com/MystenLabs/sui/blob/main/doc/paper/sui.pdf}},
  note         = {Accessed 2025-11}
}

@misc{solanaTurbine,
  title        = {Solana Turbine: Block Propagation Protocol},
  author       = {{Solana Labs}},
  year         = {2023},
  howpublished = {\url{https://docs.solana.com/cluster/turbine}},
  note         = {Accessed 2025-11}
}

@misc{albassam2021lazyledger,
  title = {LazyLedger: A Distributed Data Availability Ledger With Client-Side Smart Contracts},
  author = {Mustafa Al-Bassam},
  year = {2019},
  note={\url{https://arxiv.org/abs/1905.09274}}, 
}

@inproceedings{chan2023simplex,
  title={Simplex consensus: A simple and fast consensus protocol},
  author={Chan, Benjamin Y and Pass, Rafael},
  booktitle={TCC},
  pages={452--479},
  year={2023},
}

@misc{abraham2025latency,
  author       = {Ittai Abraham and Yuval Efron and Ling Ren},
  title        = {The Latency Cost of Censorship Resistance},
  year         = {2025},
  howpublished = {Manuscript},
  note         = {November 2025}
}

@inproceedings{cachin2005avid,
  author    = {Christian Cachin and Stefano Tessaro},
  title     = {Asynchronous Verifiable Information Dispersal},
  booktitle = {Proceedings of the 24th IEEE Symposium on Reliable Distributed Systems (SRDS)},
  year      = {2005},
  pages     = {191--201},
  publisher = {IEEE}
}

@article{garimidi2025mcp,
  author       = {Pranav Garimidi and Joachim Neu and Max Resnick},
  title        = {Multiple Concurrent Proposers: Why and How},
  journal      = {CoRR},
  volume       = {abs/2509.23984},
  year         = {2025},
  url          = {https://arxiv.org/abs/2509.23984}
}

@inproceedings{shoup2023sing,
  title={Sing a song of Simplex},
  author={Shoup, Victor},
  booktitle = {DISC},
  year={2024},
  pages = {37:1–37:22},
}

@article{goren2025shelby,
  title={Shelby: Decentralized Storage Designed to Serve},
  author={Goren, Guy and Hariri, Andrew and Hartley, Timothy DR and Kappiyoor, Ravi and Spiegelman, Alexander and Zmick, David},
  journal={arXiv preprint arXiv:2506.19233},
  year={2025}
}

@article{danezis2025walrus,
  title={Walrus: An Efficient Decentralized Storage Network},
  author={Danezis, George and Giuliari, Giacomo and Kogias, Eleftherios Kokoris and Legner, Markus and Smith, Jean-Pierre and Sonnino, Alberto and W{\"u}st, Karl},
  journal={arXiv preprint arXiv:2505.05370},
  year={2025}
}

@misc{kniep2025solana,
  title={Solana alpenglow consensus},
  author={Kniep, Quentin and Sliwinski, Jakub and Wattenhofer, Roger},
  year={2025}
}

@article{stathakopoulou2019mir,
  title={Mir-{BFT}: High-throughput {BFT} for blockchains},
  author={Stathakopoulou, Chrysoula and David, Tudor and Vukolic, Marko},
  journal={arXiv preprint arXiv:1906.05552},
  volume={92},
  year={2019}
}

@inproceedings{giridharan2024autobahn,
  title={Autobahn: Seamless high speed {BFT}},
  author={Giridharan, Neil and Suri-Payer, Florian and Abraham, Ittai and Alvisi, Lorenzo and Crooks, Natacha},
  booktitle={SOSP},
  pages={1--23},
  year={2024}
}

@misc{eip4844,
  title        = {EIP-4844: Shard Blob Transactions (Proto-Danksharding)},
  author       = {Dankrad Feist and others},
  year         = {2023},
  howpublished = {\url{https://eips.ethereum.org/EIPS/eip-4844}},
  note         = {Accessed 2025-11}
}

@misc{buterin2020eth2,
  title        = {Ethereum 2.0 Phase 0/1/2 Roadmap (Gasper overview)},
  author       = {Vitalik Buterin},
  year         = {2020},
  howpublished = {\url{https://ethereum.org/}},
  note         = {Accessed 2025-11}
}

@inproceedings{gabizon2018hotstuff,
  title        = {HotStuff: {BFT} Consensus in the Lens of Blockchain},
  author       = {Maofan Yin and Dahlia Malkhi and Michael K. Reiter and Guy Golan-Gueta and Ittai Abraham},
  booktitle    = {PODC},
  year         = {2019},
  pages = {347--356}
}

@misc{sui-code,
    title = {"Sui"},
    author = {{Mysten Labs}},
    howpublished = {\url{https://github.com/mystenlabs/sui}},
    year = {2025}
}
    
    \end{document}